\documentclass[aps,showpacs,nofootinbib,amsmath,amssymb,preprint]{revtex4}
\usepackage{bm}
\usepackage{graphicx}
\usepackage{amsthm}

\newtheorem{lemma}{Lemma}
\newtheorem{proposition}{Proposition}

\begin{document}

\title{
Stability analysis of black holes by the $S$-deformation method for coupled systems
}

\author{
Masashi Kimura${}^{1}$
and
Takahiro Tanaka${}^{2,3}$,
}

\affiliation{
${}^{1}$CENTRA, Departamento de F\'{\i}sica, Instituto Superior T\'ecnico, Universidade de Lisboa, Avenida~Rovisco Pais 1, 1049 Lisboa, Portugal
\\
${}^{2}$Department of Physics, Kyoto University, Kyoto 606-8502, Japan
\\
${}^{3}$
Yukawa Institute for Theoretical Physics, Kyoto University, Kyoto 606-8502, Japan
}

\date{\today}
\preprint{KUNS-2735}
\preprint{YITP-18-97}

\begin{abstract}
We propose a simple method to prove the linear mode stability of a black hole
when the perturbed field equations take the form of a system of coupled Schr\"odinger equations.
The linear mode stability of the spacetime is guaranteed by the 
existence of an appropriate $S$-deformation.  
Such an $S$-deformation is related to the Riccati transformation 
of a solution to the Schr\"odinger system with zero energy.
We apply this formalism to some examples and numerically study their stability.
\end{abstract}

\maketitle

\section{Introduction}

To understand the physical properties of black holes, one typically studies their linearised field equations and the motion of test particles. The background spacetime is required to be stable for the perturbative approximation to remain valid. 
Stability is also important for black hole formation
because unstable solutions are not realised as the final states of gravitational collapse.
If the background spacetime is highly symmetric,
such as for a spherically symmetric static black hole, the perturbed field equations can be written as an expansion of mode functions, leading to a system of ordinary differential equations.
When there is only a single degree of freedom, 
or each degree of freedom is decoupled, the perturbed field equations usually 
reduce to a single master equation in the form of the Sch\"odinger equation~\cite{Regge:1957td, Vishveshwara:1970cc, Zerilli:1970se, Zerilli:1971wd, Kodama:2003jz, Ishibashi:2003ap, Kodama:2003kk, Dotti:2004sh, Dotti:2005sq, Gleiser:2005ra, Takahashi:2010ye, Takahashi:2009dz, Takahashi:2009xh}
\begin{align}
\left(-\frac{d^2}{dx^2} + V\right)\Phi = \omega^2 \Phi,
\label{singleschrodingereq}
\end{align}
where we have expressed the time dependence of the perturbed fields in terms of the modes $e^{-i\omega t}$.
The existence of $\omega^2 < 0$ mode for Eq.~\eqref{singleschrodingereq} 
corresponds to an exponentially growing mode.
The $S$-deformation method was used to show the non-existence of such modes in~\cite{Kodama:2003jz, 
Ishibashi:2003ap, Kodama:2003kk, Dotti:2004sh, Dotti:2005sq, Gleiser:2005ra, 
Takahashi:2010ye, Takahashi:2009dz, Takahashi:2009xh, Beroiz:2007gp, Takahashi:2010gz}.
In the $S$-deformation method, the existence of a function $S$ that is continuous everywhere 
and satisfies
$V - S^2 + dS/dx \ge 0$ 
implies the mode stability of the spacetime~\cite{Kodama:2003jz, Ishibashi:2003ap, Kodama:2003kk}.
Recently, it was shown that
we can construct a regular solution of
\begin{align}
V - S^2 + \frac{dS}{dx} = 0,
\end{align}
if the spacetime is stable~\cite{Kimura:2017uor, Kimura:2018eiv}.

With two or more degrees of freedom, the perturbed field equations form a coupled system of ordinary differential equations.
In some cases, the perturbed field equations take the form of 
systems of coupled Schr\"odinger equations~\cite{Sarbach:2001mc, Winstanley:2001bs, Baxter:2015gfa, Molina:2010fb, Nishikawa:2010zg, Cardoso:2018ptl}, 
where the form of the equation is the same as Eq.~\eqref{singleschrodingereq} 
but the potential is a matrix and the wave function has multi components.
In this paper, 
to study the mode stability of such systems, 
we extend the formalism in the previous works~\cite{Kimura:2017uor, Kimura:2018eiv} to coupled systems.
We also discuss the relation between our formalism and the nodal theorem in the pioneering work~\cite{Amann:1995}
on this topic.

This paper is organized as follows.
In Sec.~\ref{sec:2}, we develop a formalism to prove 
stability of the coupled systems based on the $S$-deformation method.
In Sec.~\ref{sec:application}, we apply our formalism to some examples,
and numerically study their stability.
Sec.~\ref{sec:summary} is devoted to summary and discussion.
We also provide several appendices. 
In this paper, we consider that the effective potential is a Hermitian matrix, but
we also show that it can be considered as real symmetric matrix problem in 
Appendix.~\ref{appendix:vermitianvasrealsymm}.
In Appendix.~\ref{appendix:riccatitr}, we discuss basic properties of the Riccati transformation
for a system of coupled Schr\"odinger equations.
In Appendix.~\ref{appendix:compactpotential}, we give a proof of the existence of $S$-deformation
for the case with a compact support potential if the system is stable.
In Appendix.~\ref{appendix:robustness}, we discuss the 
robustness of the $S$-deformation method, {\it i.e.,}
the reason why we can find a regular $S$-deformation without fine-tuning.
In Appendix.~\ref{appendix:positivedefinite}, 
we give a proof of the existence of a regular $S$-deformation for a positive definite effective potential.
In Appendix.~\ref{approximatesolutionsdef}, we give an explicit form of $S$-deformation 
near the horizon for two degrees of freedom case.
In Appendix.~\ref{appendix:zeromode}, we show that the existence of two different $S$-deformations 
implies the non-existence of the zero mode under some assumption.

\section{Simple test for stability of black hole}
\label{sec:2}
\subsection{The $S$-deformation method for coupled systems}
Consider the case where the perturbed field equations take 
the form of a system of coupled Schr\"odinger equations
\begin{align}
-\frac{d^2}{dx^2}\bm{\Phi} + \bm{V}\bm{\Phi}= \omega^2 \bm{\Phi} =: E \bm{\Phi},
\label{multischrodingereq}
\end{align}
where $\bm{V}$ is an $n \times n$ Hermitian matrix~\footnote{
We consider the case where the coupling term in the perturbed Lagrangian 
takes $\bm{\Phi}^\dag \bm{V} \bm{\Phi}$. If $\bm{V}$ is a Hermitian matrix, 
this term is real.
When we consider real symmetric $\bm{V}$ and real $\bm{\Phi}$,
the following discussion in this paper holds
just by taking real parts. 
}, and $\bm{\Phi}$ is an $n$-component vector.
We assume that $n$ also corresponds to the number of physical degrees of freedom.
In this paper, we assume that the domain of $\bm{V}$ is $-\infty < x < \infty$,
and $\bm{V}$ is piecewise continuous and bounded.
For any $n\times n$ matrix ${\bm S}$,
we can show the relation
\begin{align}
-\frac{d}{dx} \left[
\bm{\Phi}^\dag \frac{d\bm{\Phi} }{dx}
+
\bm{\Phi}^\dag \bm{S} \bm{\Phi}
\right]
+
\left(\frac{d\bm{\Phi}^\dag}{dx} + \bm{\Phi}^\dag \bm{S}  \right)
\left(\frac{d\bm{\Phi}}{dx} +\bm{S} \bm{\Phi}   \right)
+ 
\Phi^\dag \left[\bm{V} + \frac{d\bm{S}}{dx} - \bm{S}^2 \right] \Phi = E |\bm{\Phi}|^2,
\notag
\end{align}
where $\dag$ denotes the Hermitian conjugate.
If ${\bm S}$ is Hermitian and its components are continuous functions,\footnote{
We also assume that
$d\bm{S}/dx$ can be defined piecewise continuously, 
and $d\bm{S}/dx$ is bounded at the discontinuity points.}
the equation
\begin{align}
& -
\left[
\bm{\Phi}^\dag\frac{d\bm{\Phi}}{dx} + \bm{\Phi}^\dag \bm{S}\bm{\Phi}
\right]_{- \infty}^{\infty}
+
\int dx
\left[
\left|\frac{d\bm{\Phi}}{dx} + \bm{S}\bm{\Phi}\right|^2 + 
\bm{\Phi}^\dag\left(\bm{V} + \frac{d\bm{S}}{dx} - \bm{S}^2 \right)\bm{\Phi} 
\right]
=
E \int dx\left|\bm{\Phi} \right|^2,
\label{multisdef1}
\end{align}
holds. We consider the boundary condition such that the boundary term in Eq.~\eqref{multisdef1} vanishes.
We can see that the deformed potential for the coupled system is 
\begin{align}
\tilde{\bm{V}} := \bm{V} + \frac{d\bm{S}}{dx} - \bm{S}^2.
\label{eq:vsdeformation}
\end{align}
This is an extension of $S$-deformation method~\cite{Kodama:2003jz, Ishibashi:2003ap, Kodama:2003kk},
and we also refer to $\tilde{\bm{V}}$ 
as an $S$-deformation of the potential $\bm{V}$ in this paper.
If there exists a continuous $\bm{S}$ which gives $\tilde{\bm{V}} = \bm{K}^\dag \bm{K}$, 
since $\bm{\Phi}^\dag \tilde{\bm{V}} \bm{\Phi} = |\bm{K}\bm{\Phi}|^2 \ge 0$,
we can say the non-existence of the negative energy bound state, 
{\it i.e.}, the non-existence of the exponentially growing mode in time.
In~\cite{Aybat:2010sn, Kimura:2018nxk}, this method was used for stability analysis.

\subsection{stability of the coupled system}
Similarly to the single mode case~\cite{Kimura:2017uor, Kimura:2018eiv}, we consider the condition $\tilde{\bm{V}} = 0$, {\it i.e.,}
\begin{align}
\bm{V} + \frac{d\bm{S}}{dx} - \bm{S}^2 = 0,
\label{sdefeqmulti}
\end{align}
for the coupled system.
The existence of a continuous $\bm{S}$ as a solution of Eq.~\eqref{sdefeqmulti}
is a sufficient condition for the stability of the spacetime from Eq.~\eqref{multisdef1}.

According to the nodal theorem in~\cite{Amann:1995},\footnote{
Note that this theorem holds when $\bm{V}$ is Hermitian 
while a real symmetric potential is assumed in~\cite{Amann:1995}.
This is because the Hermitian case can be considered as a real symmetric problem as shown in 
Appendix.~\ref{appendix:vermitianvasrealsymm}.}
for a set of the solutions $\{\bm{\Phi}_i\}$ ($i = 1,2,\ldots n$) 
of the Schr\"odinger equation with $E = 0$ 
with the boundary condition such that
$\bm{\Phi}_i|_{x = L} = 0$ and $(d\bm{\Phi}_i/dx)|_{x = L} = \bm{v}_i$, 
where $\bm{v}_i$ are linearly independent constant vectors,
the necessary and sufficient condition for the non-existence of the negative energy bound state 
for Eq.~\eqref{multischrodingereq}
is that $\det(\bm{Y})$
does not have a zero except at $x = L$ 
if 
$L$ is sufficiently large (or sufficiently large negative),
where $\bm{Y}$ is defined as
\begin{align}
\bm{Y}:= (\bm{\Phi}_1, \bm{\Phi}_2, \ldots, \bm{\Phi}_n).
\end{align}
In~\cite{Winstanley:2001bs, Garaud:2007ti, Garaud:2010ng, Baxter:2015gfa, 
Chen:2017diy}, the nodal theorem~\cite{Amann:1995} was used for stability analysis.

If we assume that this statement holds even when 
we take the limit $L \to \infty$ (or $L \to -\infty$),
{\it i.e.,} 
when we take $\{\bm{\Phi}_i\}$ as a set of 
$n$-decaying modes at $x \to -\infty$ (or $x \to \infty$), 
 $\det(\bm{Y}) \neq 0$ except at infinity when the spacetime is stable, and 
$\bm{Y}^{-1}$ does not diverge at any finite point.
In this case, we can construct a continuous $\bm{S}$ as a Riccati transformation 
by 
\begin{align}
\bm{S}:= -\frac{d\bm{Y}}{dx}\bm{Y}^{-1}.
\label{sfromyzeroenergy}
\end{align}
We can easily check that $\bm{S}$ satisfies Eq.~\eqref{sdefeqmulti} and 
becomes a Hermitian matrix
(see Appendix.~\ref{appendix:riccatitr} for
basic properties of the Riccati transformation 
of solutions to systems of coupled Schr\"odinger equations).
This discussion suggests that 
there exists 
a continuous solution $\bm{S}$ for Eq.~\eqref{sdefeqmulti} 
if there does not exist a bound state with $E \le 0$.
In Appendix~\ref{appendix:compactpotential}, we show that 
this is correct for the case with compact support potential
(or rapidly decaying potential at $x \to \pm \infty$).
Also, in a single mode case, a proof was given under weaker conditions~\cite{Kimura:2018eiv}.

We introduce $\bm{Y}_L$ and $\bm{Y}_R$ to denote $\bm{Y}$ which 
are constructed from linearly independent decaying (or constant) modes at $x \to -\infty$
and $x \to \infty$, respectively,
and the corresponding $\bm{S}$ as $\bm{S}_L = -(d\bm{Y}_L/dx)\bm{Y}_L$ 
and $\bm{S}_R = -(d\bm{Y}_R/dx)\bm{Y}_R$, respectively.
As discussed in Appendix~\ref{appendix:robustness},
we can show that
the general regular $\bm{S}$ for the compact support potential satisfies the property that
all eigenvalues of $\bm{S}_R - \bm{S}$ and $\bm{S} - \bm{S}_L$ are non-negative.
Also, if we solve Eq.~\eqref{sdefeqmulti} with the boundary conditions that
guarantee $\bm{S}$ to be Hermitian
and all eigenvalues of $\bm{S}_R - \bm{S}$ and $\bm{S} - \bm{S}_L$ to be non-negative,
the solution becomes regular for a stable spacetime.
We expect that these properties hold even for non-compact potentials with sufficiently 
rapid convergence at $x \to \pm \infty$.

In practice, if the potential $\bm{V}$ is positive definite at large $x$, 
as discussed in Appendix.~\ref{appendix:positivedefinite} 
where we show the existence of a regular $S$-deformation for a positive definite potential,
we conjecture that $\bm{S} = 0$ at a large $x$ is
an appropriate boundary condition in solving Eq.~\eqref{sdefeqmulti} to obtain a regular $\bm{S}$.

\subsection{$\bm{S}$ is bounded if ${\rm Tr}(\bm{S})$ and $\bm{V}$ are bounded}
\label{sec:boundedtrs}

When $\bm{S}$ is divergent at a point, one of its eigenvalues is also divergent there.
Let $\bm{e}$ be a unit eigenvector of $\bm{S}$ with the eigenvalue $\lambda$.
Then, $\lambda$ satisfies the equation
\begin{align}
\frac{d\lambda}{dx} = \lambda^2 - \bm{e}^\dag \bm{V} \bm{e}.
\end{align}
Note that $\bm{e}$ is not a constant vector.
For a bounded $\bm{V}$, $\lambda$ can be divergent only when $d\lambda/dx \simeq \lambda^2$.
Since the solution of $d\lambda/dx = \lambda^2$ is $\lambda = -1/(x-c)$, 
the eigenvalue can be divergent only at a finite point.
Thus, if $\bm{V}$ are bounded, $\bm{S}$ can be divergent only at a finite point.

If ${\rm Tr} \bm{S}$ is bounded above and below, 
$\det (\bm{Y})$ does not have zero except at infinity
since $ d (\ln \det (\bm{Y}) )/dx = {\rm Tr}(\bm{Y}^{-1}d\bm{Y}/dx) = {\rm Tr}(\bm{S})$.
This implies that $\bm{S}$ is regular if ${\rm Tr} \bm{S}$ is bounded above and below.

\section{Application: numerical calculations for coupled systems}
\label{sec:application}
In this section, we consider to solve Eq.~\eqref{sdefeqmulti} numerically and
find a regular $S$-deformation of a coupled system.
As examples,
we apply our formalism to the systems discussed in~\cite{Molina:2010fb, Nishikawa:2010zg}.

\subsection{Schwarzschild black hole in dynamical Chern-Simons gravity}
\label{sec:dcs}
The odd parity metric perturbations coupled to a perturbed massless scalar field in 
dynamical Chern-Simons gravity~\footnote{
The action of this theory is 
$S = \int d^4 x \sqrt{-g}[
\kappa R - (\vartheta/4)R^{\mu \nu \rho \sigma}{}^{\ast}\!R_{\mu \nu \rho \sigma}
-(\beta/2)(\partial_\mu \vartheta \partial^\mu \vartheta)
]$, where $\kappa$ is the gravitational constant, and $\beta$ is the coupling constant~\cite{Molina:2010fb, Kimura:2018nxk}.
}
around the Schwarzschild black hole 
\begin{align}
ds^2 = -fdt^2 + f^{-1}dr^2 + r^2 (d\theta^2 + \sin^2\theta d\phi^2),
\end{align}
with $f = 1-2M/r$
reduce to the master equations~\cite{Molina:2010fb, Kimura:2018nxk}
\begin{align}
-\frac{d^2}{dx^2} \Phi_1 + V_{11}\Phi_1 + V_{12}\Phi_2 &= E \Phi_1,
\label{mastereqdcs1}
\\
-\frac{d^2}{dx^2} \Phi_2 + V_{12}\Phi_1 + V_{22}\Phi_2 &= E \Phi_2,
\label{mastereqdcs2}
\end{align}
where $d/dx = fd/dr$.\footnote{
The relation among $\Phi_1, \Phi_2$ and perturbed quantities can be seen in~\cite{Molina:2010fb, Kimura:2018nxk}.
}
The effective potentials are given by
\begin{align}
V_{11} &= f\left[\frac{\ell(\ell+1)}{r^2} - \frac{6M}{r^3}\right],
\\
V_{12} &= f \frac{24 M \sqrt{\pi (\ell+2)(\ell + 1)\ell (\ell-1)}}{\sqrt{\beta}r^5},
\\
V_{22} &= f \left[
\frac{\ell(\ell + 1)}{r^2}
\left(1 + \frac{576\pi M^2}{\beta r^6}\right)
+
\frac{2 M}{r^3}
\right],
\end{align}
where $\beta$ is the coupling constant in dynamical Chern-Simons gravity.
We note that the $\beta \to \infty$ limit corresponds to the general relativity case.
Recently, the stability of this system was proven analytically for $\beta \ge 0$~\cite{Kimura:2018nxk}.
We apply our formalism to this system as a test problem.

\begin{figure}[thbp]
\includegraphics[width=0.98\linewidth,clip]{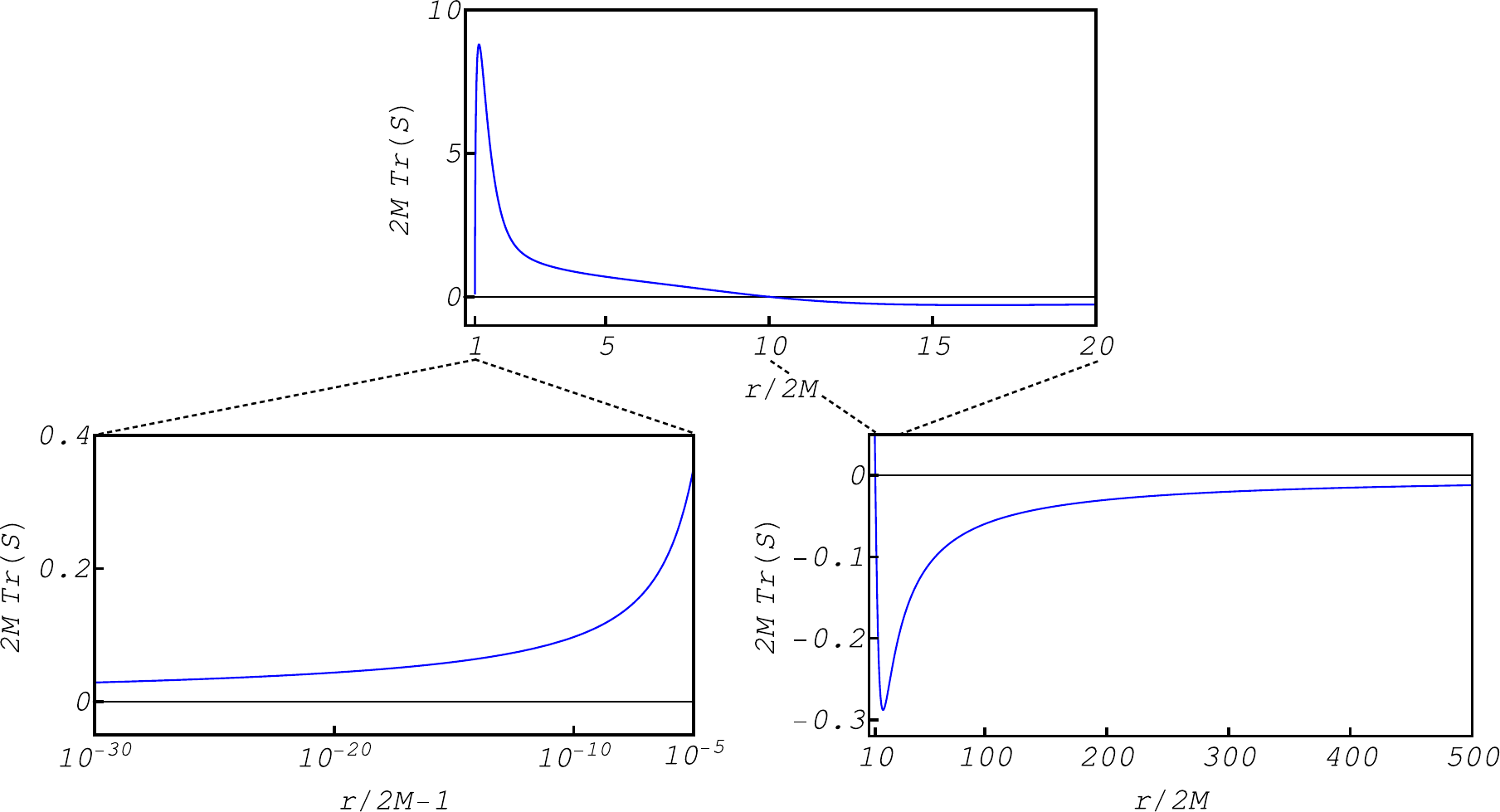}
 \caption{
The numerical solution ${\rm Tr}({\bm{S}})$ for $\ell =2, \beta M^4 = 1/10$ 
with the boundary condition $r_{\rm ini}/(2M) = 10$ and $\bm{S}|_{r_{\rm ini}} = 0$.
The figures are the profile of ${\rm Tr}({\bm{S}})$ (upper),
near horizon region (lower left), far region (lower right).
}
\label{figSdefdcs}
\end{figure}
Since $\bm{V}$ is real symmetric, we consider a real symmetric $\bm{S}$,
then Eq.~\eqref{sdefeqmulti} becomes
\begin{align}
f\frac{d S_{11}}{dr} &= S_{11}^2 + S_{12}^2  - V_{11},
\label{s11eq}
\\
f\frac{d S_{12}}{dr} &= S_{12}(S_{11} + S_{22}) - V_{12},
\label{s12eq}
\\
f\frac{d S_{22}}{dr} &= S_{22}^2 +S_{12}^2 - V_{22}.
\label{s22eq}
\end{align}
We solve these equations numerically,\footnote{
We used the function {\tt NDSolve} in {\it Mathematica} for the numerical calculation
and set the parameter {\tt WorkingPrecision} to $30$.
} and 
we plot ${\rm Tr}(\bm{S}) = S_{11} + S_{22}$ in Fig.~\ref{figSdefdcs}
for $\ell = 2, \beta M^4= 1/10$ case.
We adopt the boundary condition such that $\bm{S} = 0$ at $r/(2M) = r_{\rm ini}/(2M) = 10$,
and solve the equations in the domain 
$r_L \le r \le r_R$, with $r_L/(2M) = 1 + 10^{-30}$ and $ r_R/(2 M) =  500$.
Since $r$ is not an appropriate coordinate to  numerically solve the equations~\eqref{s11eq}-\eqref{s22eq}
in the near horizon region $r_L \le r \le r_{\rm ini}$, 
we use an alternative coordinate $\bar{x} = 2 M \ln(r/(2M) -1)$, then 
$d/d\bar{x} = (r/(2M)-1)d/dr$ and $r = 2 M(1 + e^{\bar{x}/(2M)})$.
Figure~\ref{figSdefdcs} shows that ${\rm Tr}(\bm{S})$ is continuous and bounded,
and this implies that $\bm{S}$ is regular (see Sec.~\ref{sec:boundedtrs}), {\it i.e.,} the spacetime is stable against the $\ell = 2$ mode
perturbation in $\beta M^4= 1/10$ case.
We also report that we can find a regular $\bm{S}$ 
even if we change $r_{\rm ini}$ to other values, {\it e.g.,} $r_{\rm ini}/(2M) = 50, 100, 200$,
like the single mode case~\cite{Kimura:2017uor}.

\begin{figure}[thbp]
\includegraphics[width=0.45\linewidth,clip]{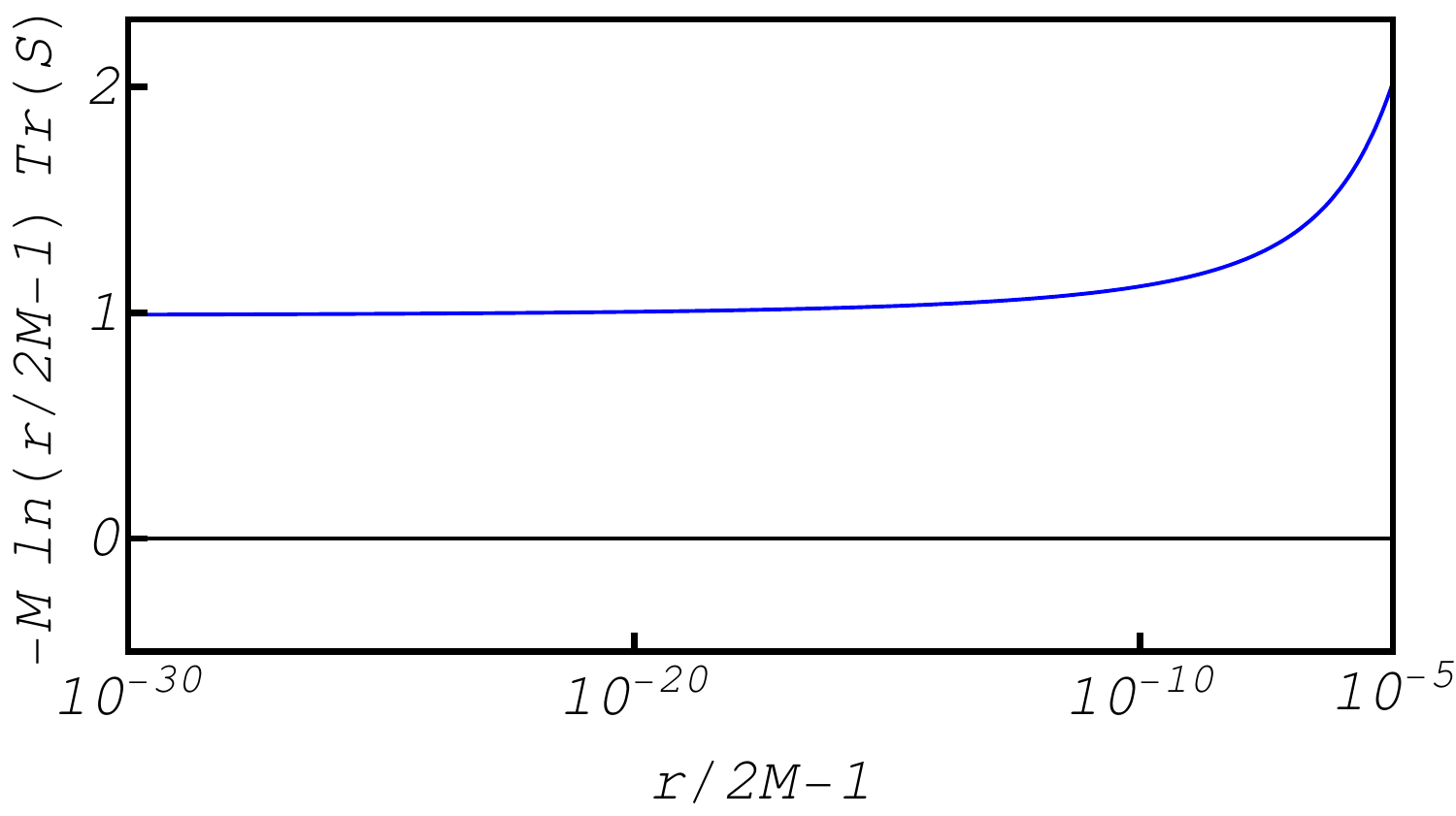}
~
\includegraphics[width=0.49\linewidth,clip]{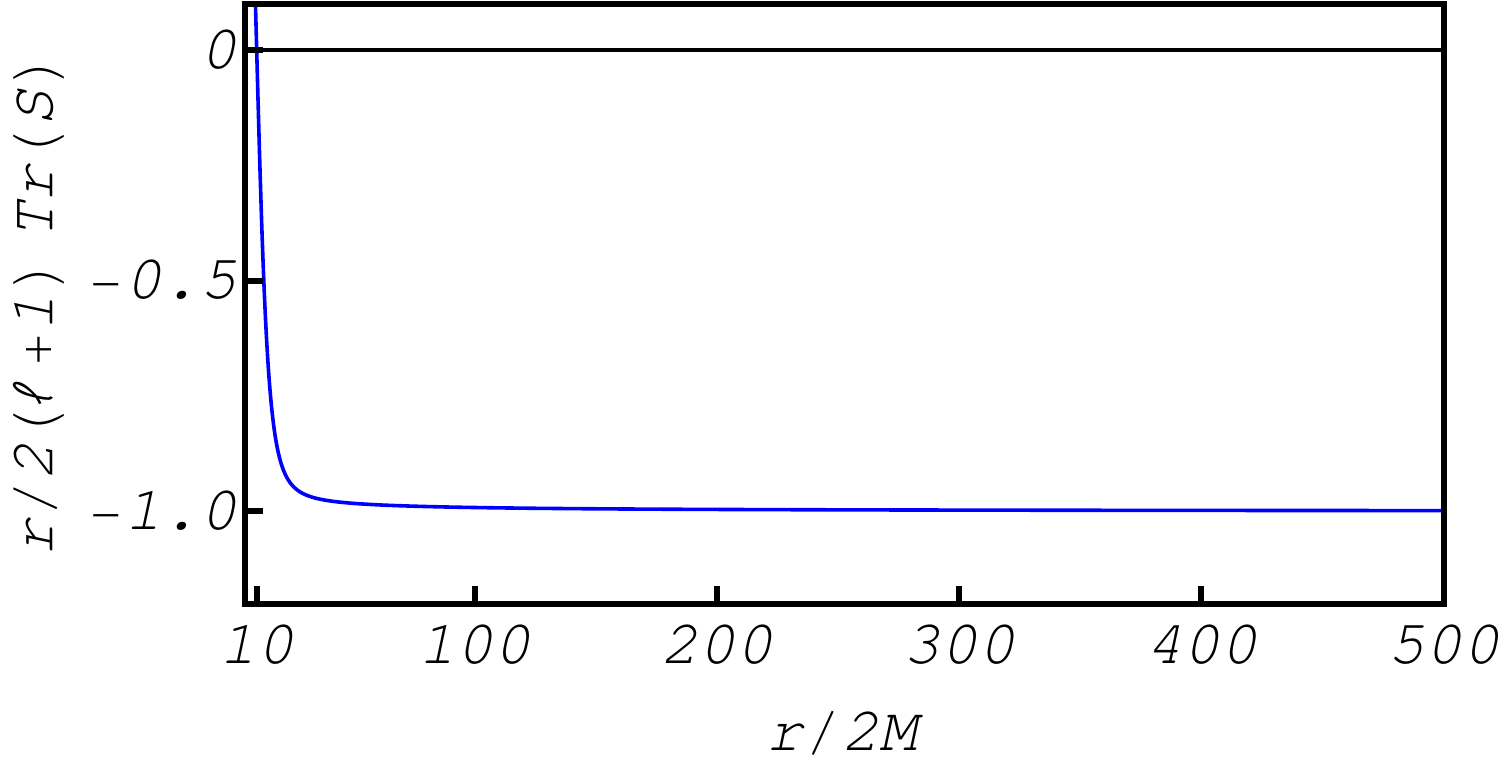}
 \caption{
Normalized ${\rm Tr}({\bm{S}})$ in near horizon region (left) 
and far region (right) by their asymptotic behaviors
${\rm Tr}({\bm{S}}) \simeq -2/\bar{x} = -1/(M \ln(r/(2M) - 1))$ and 
${\rm Tr}({\bm{S}}) \simeq -2(\ell+1)/r$, respectively.
The parameters and the boundary condition are the same as in Fig.~\ref{figSdefdcs}.
}
\label{figSdefdcsnormalized}
\end{figure}

\begin{figure}[thbp]
\includegraphics[width=0.5\linewidth,clip]{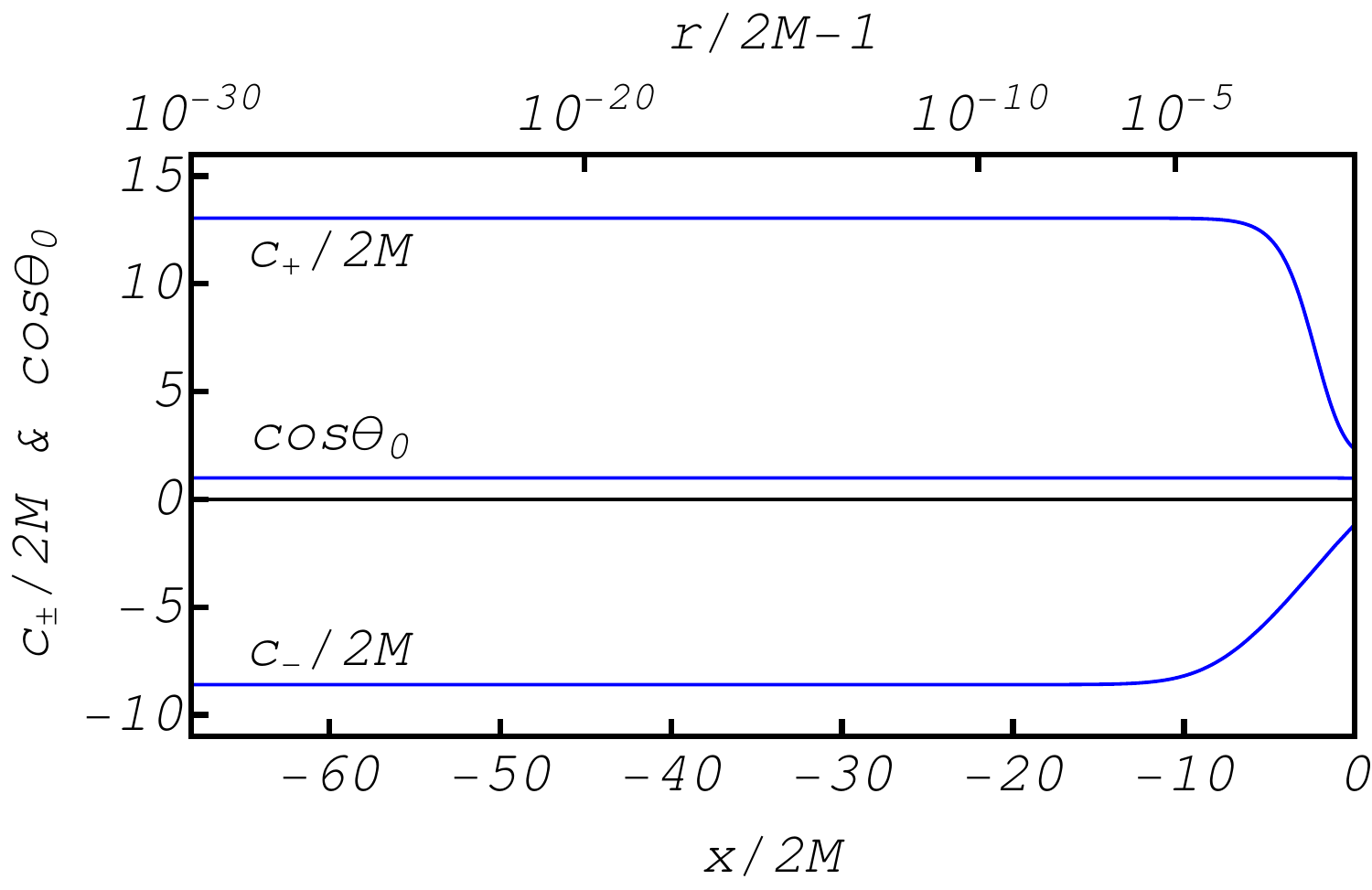}
 \caption{Quantities which asymptote to constants of the asymptotic solutions near the horizon.
The three curves $c_\pm$ and $\cos\theta_0$ correspond to the 
right hand side of Eqs.~\eqref{eq:cpm} and \eqref{eq:thetazero}.
The horizontal axis is the tortoise coordinate $x = r + 2M \ln(r/(2M)-1)$.
Note that $\cos\theta_0$ is almost unity, so $S_{12} \simeq 0$ near the horizon.}
\label{figdcsconst}
\end{figure}
The left figure in Fig.~\ref{figSdefdcsnormalized} shows that
the solution asymptotes to the approximate solution in the near horizon region, 
{\it i.e.,} the solution of Eq.~\eqref{sdefeqmulti} with vanishing potential
(see Appendix.~\ref{approximatesolutionsdef} where we show the explicit form of the general
$S$-deformation for $n = 2$ case when the potential vanishes).
In Fig.~\ref{figdcsconst}, we plot the right hand side of Eqs.~\eqref{eq:cpm} and \eqref{eq:thetazero}.
This also shows that the numerical solution can be matched with the approximate solution in $x \lesssim -10$,
and we can see that $x < c_\pm$ are satisfied there.

In far region, since the potential is not rapidly decaying as $r \to \infty$,
the solutions do not approximately satisfy Eq.~\eqref{sdefeqmulti} with vanishing potential.
However, since $|V_{12}| \ll |V_{11}|, |V_{22}|$ and
$|S_{12}| \ll |S_{11}|, |S_{22}|$ in far region,
we can expect that the system is approximately decoupled there.
The right figure in Fig.~\ref{figSdefdcsnormalized} show that the solution asymptotes to the 
decoupled solution.\footnote{
In far region, $S_{11}$ and $S_{22}$ approximately satisfies
the equations $dS_{ii}/dr = S_{ii}^2 - \ell(\ell+1)/r^2$ $(i = 1, 2)$, and the solutions are
$S_{ii} = (\ell - c_{ii} (\ell +1)r^{1+2\ell})/(r + c_{ii} r^{2(1+\ell)}) \to -(\ell+1)/r$.}

\subsection{charged squashed Kaluza-Klein black hole}
\label{subsec:sqkkbh}
The charged squashed Kaluza-Klein black hole~\cite{Ishihara:2005dp} 
is a solution of the five-dimensional Einstein-Maxwell system.
The metric and the gauge 1-form of this spacetime
are given by
\begin{align}
ds^2 &= - F dt^2 + \frac{K^2}{F}d\rho^2 + \rho^2 K^2 (d\theta^2 + \sin^2\theta d\phi^2)
+ \frac{(\rho_0 + \rho_+)(\rho_0 + \rho_-)}{K^2}(d\psi + \cos\theta d\phi)^2,
\notag
\\
A_{\mu}dx^\mu &= \frac{\sqrt{3}}{2}\frac{\rho_+\rho_-}{\rho}dt,
\notag
\end{align}
with the functions $F = (\rho - \rho_+)(\rho - \rho_-)/\rho^2, K^2 = (\rho + \rho_0)/\rho$.
The ranges of the angular coordinates are $0\le \theta \le \pi, 0\le \phi \le 2\pi, 0\le \psi \le 4\pi$.
The three parameters $\rho_\pm, \rho_0$ satisfy $\rho_+ \ge \rho_- \ge 0$ and $\rho_- + \rho_0 \ge 0$,
and the black hole horizon locates at $\rho = \rho_+$.
The relations among these parameters and the Komar mass $M$, the electric charge $Q$ and the size of the
extra dimension at infinity $r_\infty$ are
\begin{align}
M = \frac{2\pi r_\infty}{G}(\rho_+ + \rho_-), ~~ 
Q = \frac{\sqrt{3}\pi r_\infty}{G} \sqrt{\rho_+\rho_-},~~
r_\infty = 2\sqrt{(\rho_0 + \rho_+)(\rho_0 + \rho_-)},
\end{align}
where $G$ is the five-dimensional gravitational constant~\cite{Ishihara:2005dp, Kurita:2007hu}.
In the asymptotic region $\rho \to \infty$, the metric behaves
\begin{align}
ds^2 = -dt^2 + d\rho^2 + \rho^2 (d\theta^2 + \sin^2\theta d\phi^2) + (r_\infty/2)^2 (d\psi + \cos \theta d\phi)^2
+ {\cal O}(1/\rho).
\end{align}
Thus, the metric in the asymptotic region behaves the four-dimensional Minkowski spacetime with twisted $S^1$.

The gravitational and electro-magnetic perturbation around 
the charged squashed Kaluza-Klein black hole was discussed in~\cite{Nishikawa:2010zg},
and the master equations for $K=\pm1$ mode perturbations take the same form as Eqs.~\eqref{mastereqdcs1} 
and \eqref{mastereqdcs2}, 
with $d/dx = (F/K^2)d/d\rho$. 
The explicit form of the effective potential is given in~\cite{Nishikawa:2010zg}.\footnote{
Since the effective potential in~\cite{Nishikawa:2010zg} is not symmetric,
we need to consider the transformation of the master variable 
$\Phi_1 := r_{\infty} \phi_{1G}, \Phi_2 := \phi_{1E}$
so that the effective potential is real symmetric.
$V_{12}$ 
is given by multiplying $r_{\infty}$ to Eq.(52) in~\cite{Nishikawa:2010zg}.
Note that $V_{11}, V_{22}$ are same form as Eqs.(51) and (54) in~\cite{Nishikawa:2010zg}, respectively.
Since under the transformation, 
$(\phi_{1G}, \phi_{1E}) \to (\Phi_1,\Phi_2)$, 
the determinant of the effective potential does not change, 
the discussion based on using the determinant in~\cite{Nishikawa:2010zg} also does not change.
}

\begin{figure}[thbp]
\includegraphics[width=0.475\linewidth,clip]{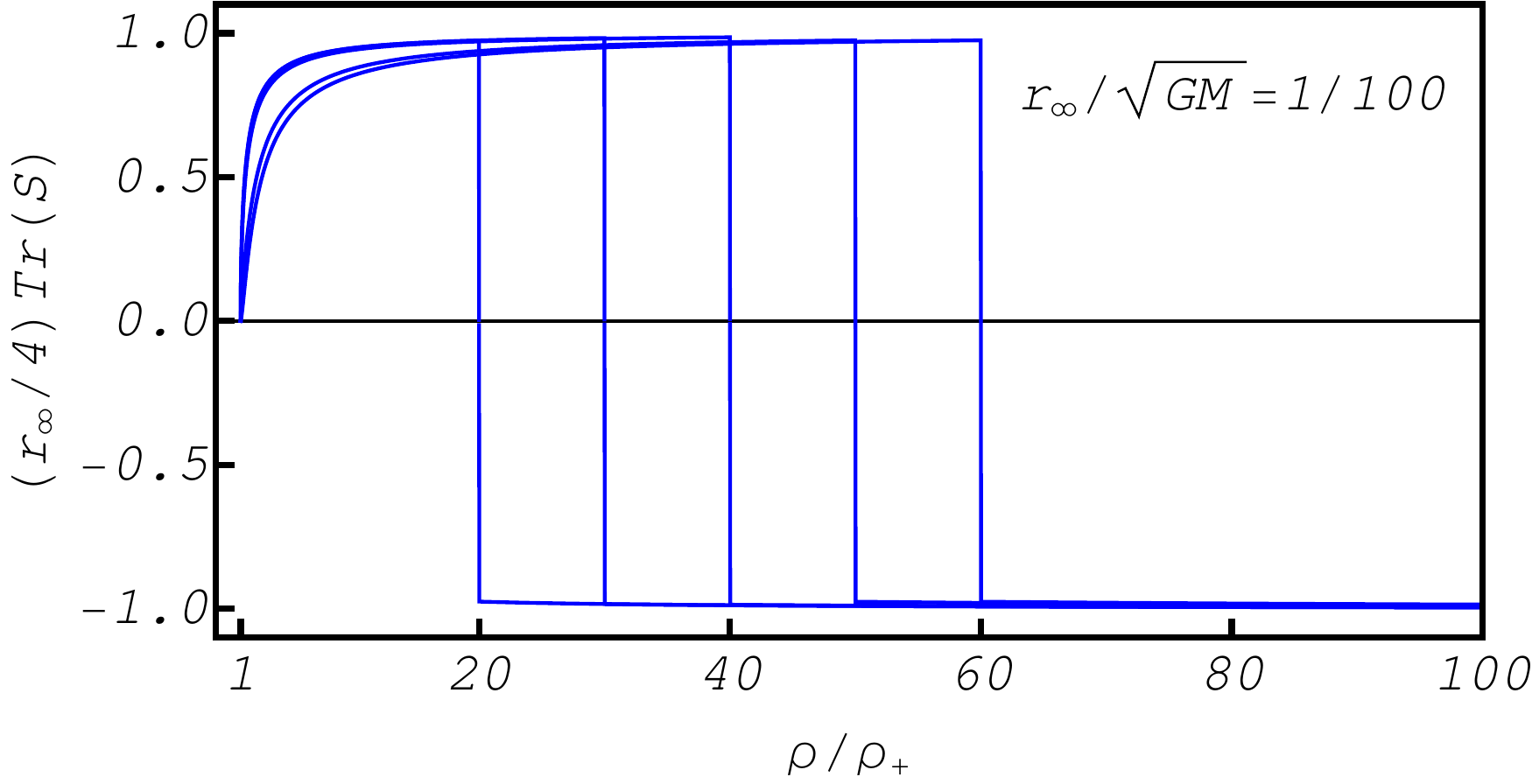}
~
\includegraphics[width=0.475\linewidth,clip]{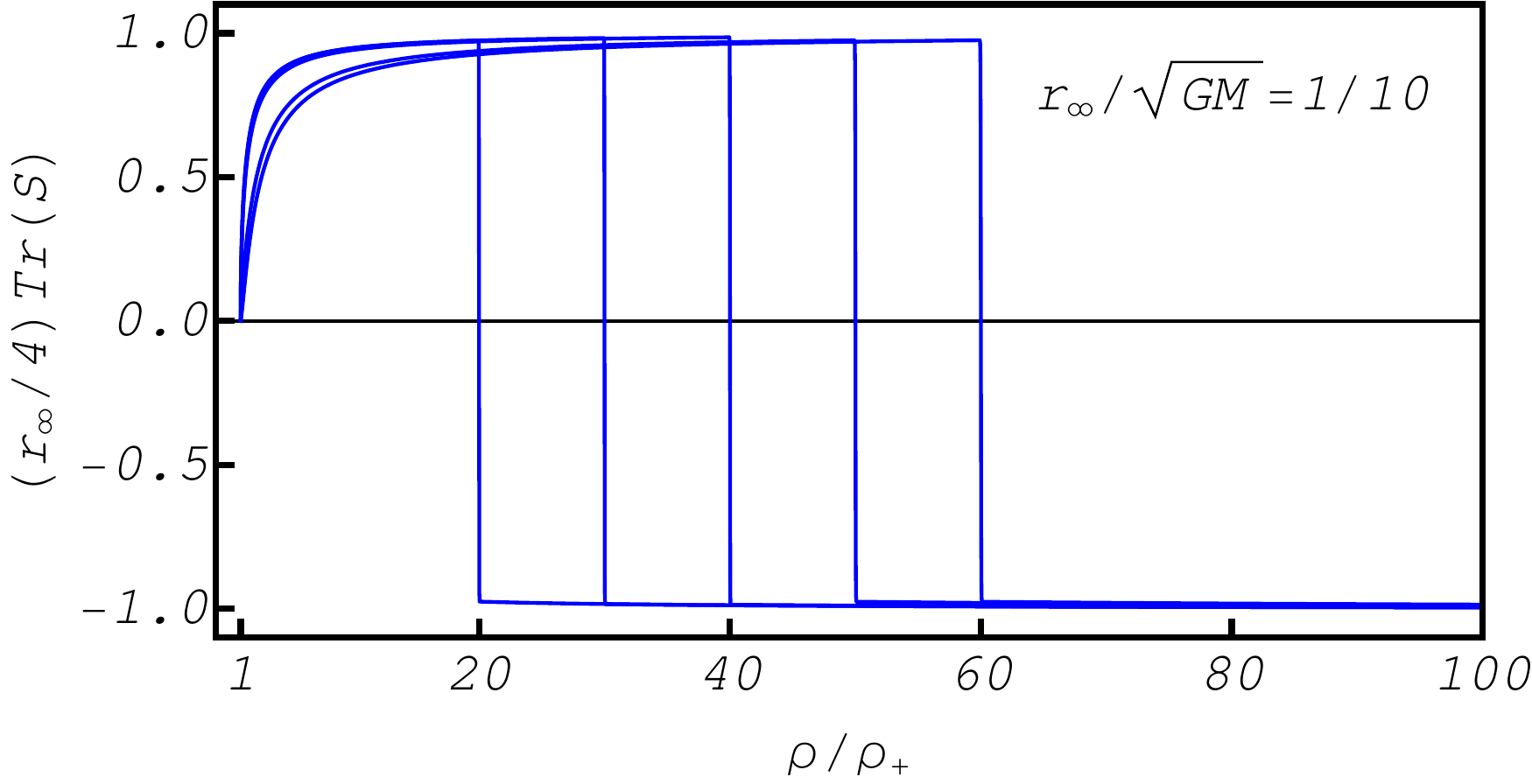}
\\
\vspace{5pt}
\includegraphics[width=0.475\linewidth,clip]{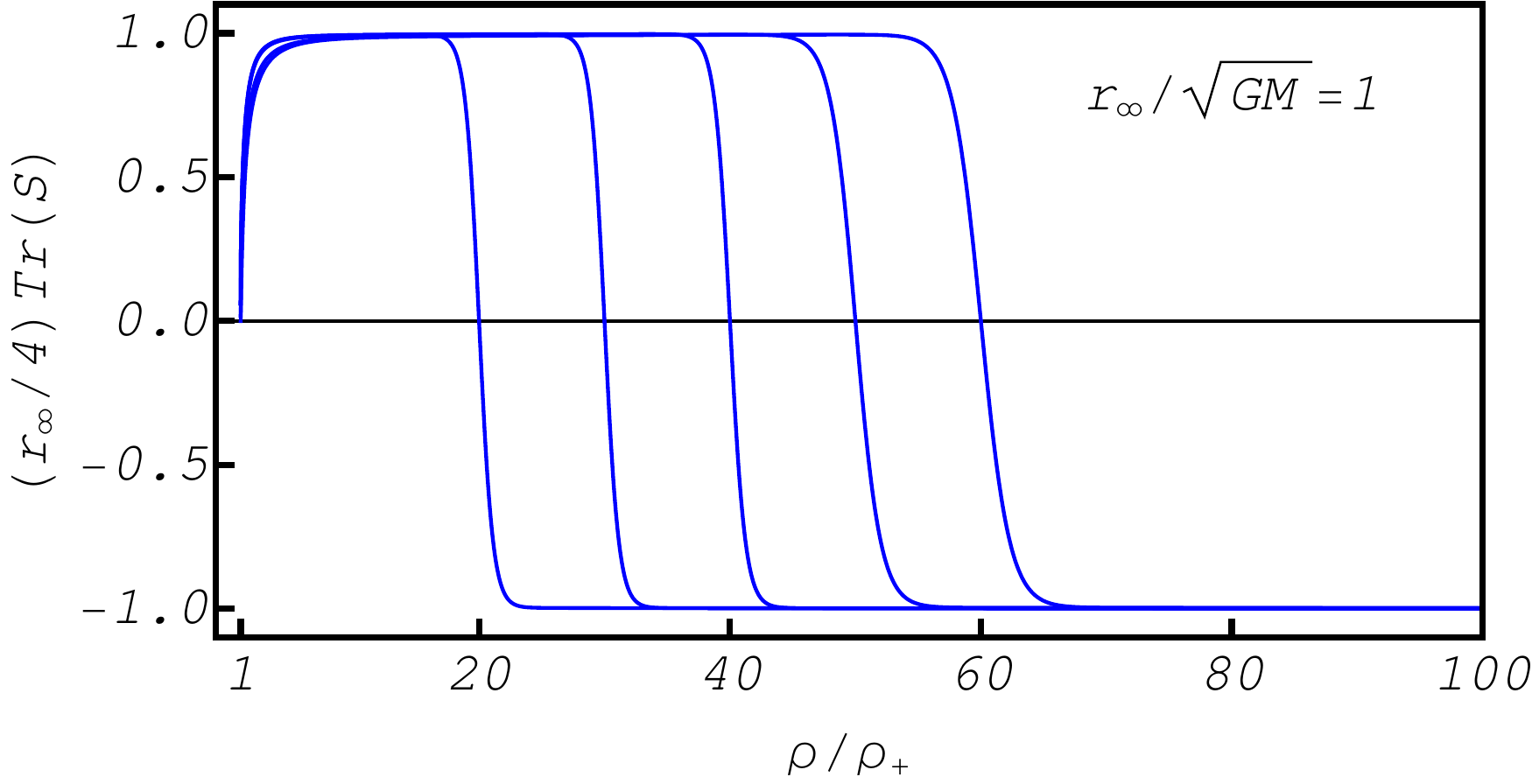}
~
\includegraphics[width=0.475\linewidth,clip]{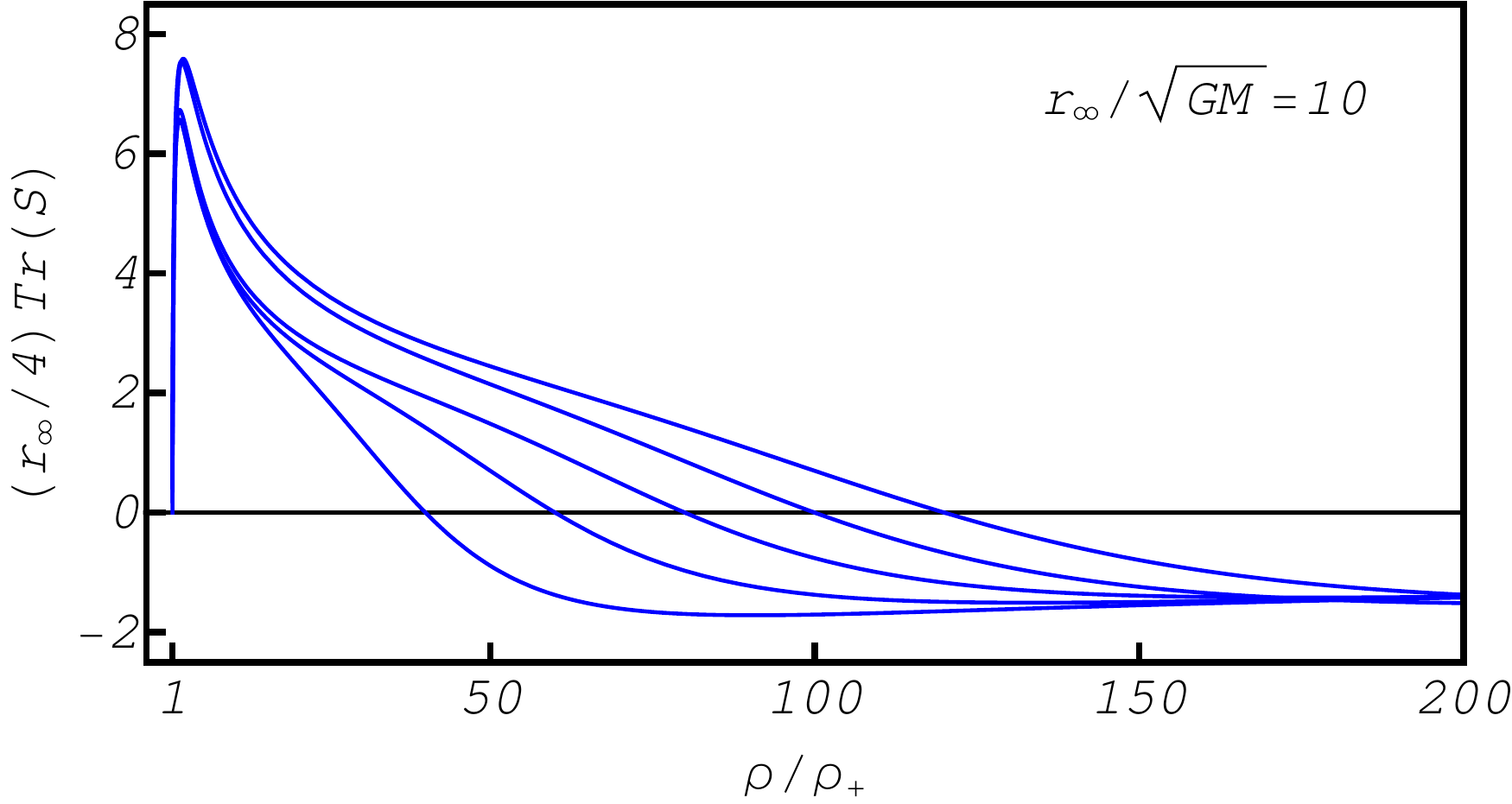}
\\
\vspace{5pt}
\includegraphics[width=0.475\linewidth,clip]{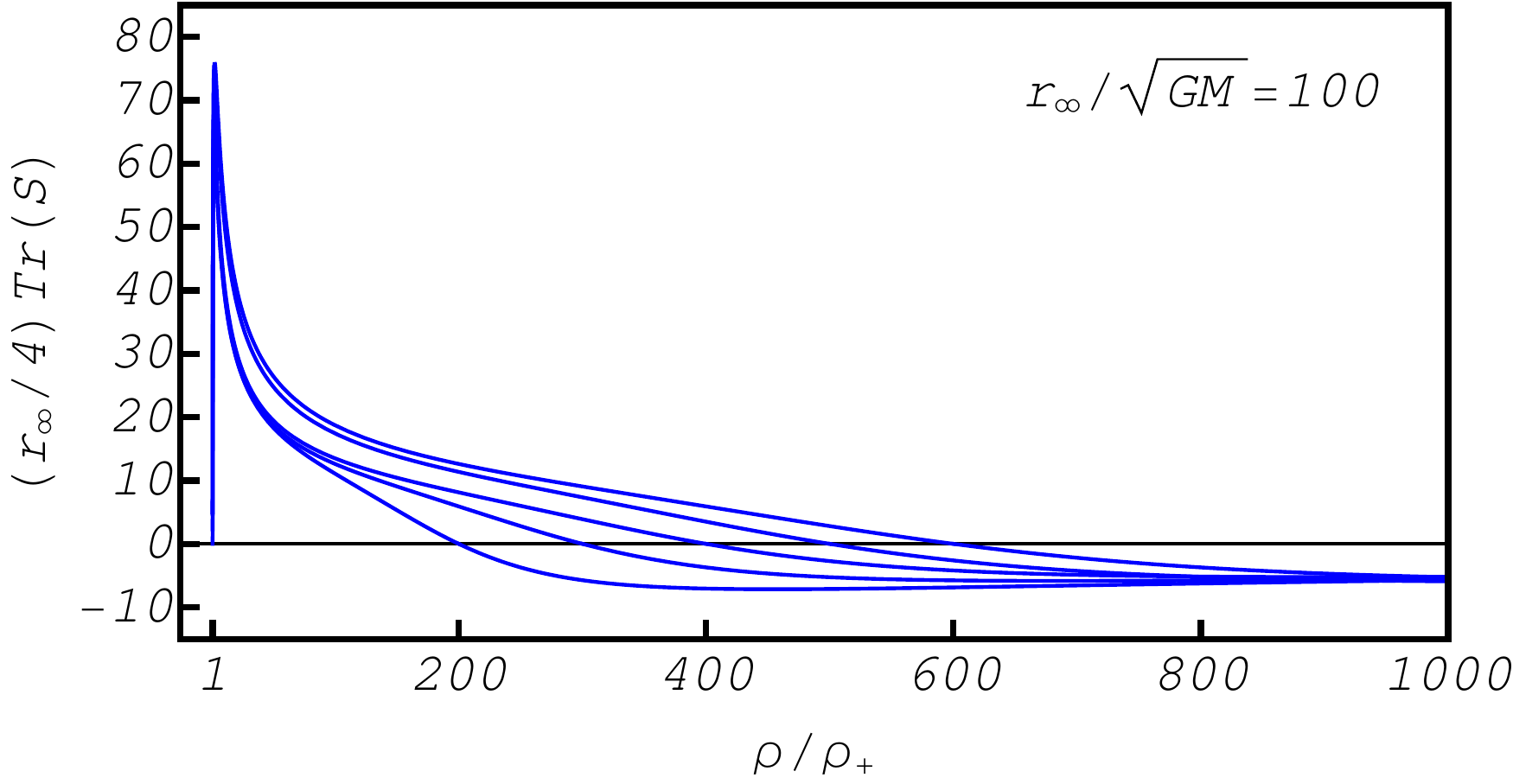}
 \caption{
${\rm Tr}({\bm{S}})$ for the numerical solution $\bm{S}$ for various parameters.
The boundary condition is $\bm{S} = 0$, and hence ${\rm Tr}(\bm{S}) = 0$, at $\rho = \rho_{\rm ini}$.
The values of $\rho_{\rm ini}$ can be read from each figure,
the corresponding charge is $Q/Q_{\rm extremal} = 0, 0.01, 0.5, 0.99, 1$ from left to right, respectively.
}
\label{figSdef}
\end{figure}
Since $\bm{V}$ is real symmetric, we consider a real symmetric $\bm{S}$. 
Then, Eq.~\eqref{sdefeqmulti} also takes the same form as Eqs.\eqref{s11eq}-\eqref{s22eq}.
At $\rho \to \infty$, 
the potential takes $V_{11} \to 4/r_{\infty}^2, V_{22} \to 4/r_{\infty}^2$ and $V_{12} \to 0$,
then the $S$-deformation takes 
$S_{11} \to -2/r_\infty, S_{22} \to -2/r_\infty, S_{12} \to 0$ and ${\rm tr}(\bm{S}) \to -4/r_\infty$.
We plot the numerical solutions ${\rm Tr}(\bm{S})$ in Fig.~\ref{figSdef} 
for the cases $r_{\infty}/\sqrt{GM} = 0.01, 0.1, 1, 10, 100$.
In each case, we plot $Q/Q_{\rm extremal} = 0, 0.01, 0.5, 0.99, 1$, respectively, 
where $Q_{\rm extremal} = \sqrt{3}M/2$ is the electric charge for the extremal case.
We adopt the boundary condition 
such that $\bm{S} = 0$ at $\rho = \rho_{\rm ini}$, 
and the values of $\rho_{\rm ini}$ can be read from Fig.~\ref{figSdef},
which correspond to $Q/Q_{\rm extremal} = 0, 0.01, 0.5, 0.99, 1,$ from left to right, respectively.
Fig.~\ref{figSdef} shows that ${\rm Tr}(\bm{S})$ is bounded,
and this implies that the spacetime is stable.\footnote{
In Fig.~1 in the erratum for \cite{Kimura:2007cr}, 
the curves for $\rho_0/\rho_+ = 30, 50, 100$ apparently looks regular, 
but in fact they are divergent near the horizon.
If we adopt the boundary condition $S = 0$ at $\rho/\rho_+ = 50$, then $S$ becomes regular everywhere.}
We should note that we can see
${\rm tr}(S) \to -4/r_\infty$ in the cases $r_{\infty}/\sqrt{M} = 10, 100,$
if we extend the plot range of the $\rho$ axis.
We just report that
we can find a regular $\bm{S}$ if we change $\rho_{\rm ini}$ to larger values.

Similarly to Sec.~\ref{sec:dcs}, we use a coordinate $\bar{x} = \rho_+ \ln(\rho/\rho_+ -1)$,
to solve Eqs.\eqref{s11eq}-\eqref{s22eq} in the near horizon region.
We checked that the right hand sides of Eqs.~\eqref{eq:cpm} and \eqref{eq:thetazero}
take constants near the horizon, and $x < c_\pm$ are satisfied there.
Note that for large $r_\infty$, $c_+$ becomes very large, {\it e.g.,} $c_+/\rho_+ \sim 10^8$, and 
then $(-x/2){\rm Tr}(\bm{S})$ becomes almost unity 
at a point extremely close to the horizon, {\it e.g.,} $\rho/\rho_+ -1 \sim 10^{-10^8}$.
In the extremal case, we use a coordinate $\tilde{x} = -\rho_+^2/(\rho - \rho_+)$
to solve Eqs.\eqref{s11eq}-\eqref{s22eq} in the near horizon region.
We note that the 
stability was already proven analytically in the extremal case~\cite{Nishikawa:2010zg}.

\section{Summary and Discussion}
\label{sec:summary}
In this paper, we extended the formalism in~\cite{Kimura:2017uor, Kimura:2018eiv} to 
coupled systems where the perturbative field equations take
the form of systems of coupled Schr\"odinger equations.
Similar to~\cite{Kimura:2017uor, Kimura:2018eiv}, the existence of $S$-deformation implies the
linear mode stability of the system.
Also, we applied our formalism to the case of 
the Schwarzschild black hole in dynamical Chern-Simons gravity
and a charged Kaluza-Klein black hole~\cite{Ishihara:2005dp, Nishikawa:2010zg},
and we showed the stability of the spacetime by finding a regular $S$-deformation function numerically.
While the dynamical Chern-Simons gravity case is a test problem because the stability 
was already shown in~\cite{Kimura:2018nxk}, 
our study of charged Kaluza-Klein black hole case gives a first numerical proof of 
the stability of the coupled mode for non-extremal cases.

For the coupled system, 
there is already a criteria, called the nodal theorem~\cite{Amann:1995}, 
for the existence and non-existence of the
negative energy bound state as an extension of the Sturm-Liouville theorem.
From the nodal theorem, we obtained a suggestion such that
there also exists a regular $S$-deformation for the coupled system if the spacetime is stable.
We showed that this is correct when the potential is of a compact support
(or rapidly decaying at infinity)
in Appendix~\ref{appendix:compactpotential}.
We mention some merit of our $S$-deformation method compared to the nodal theorem~\cite{Amann:1995}.
(i)~If we discuss the stability based on the nodal theorem,
we need to solve 
the Schr\"odinger equation from $x = L$ (or $x = -L$) with sufficiently large $L$, 
but if the spacetime is almost marginally stable, it is not trivial
how large $L$ we should consider.
If we discuss the stability based on the $S$-deformation method,
we can solve the equation from a finite point, and we do not need to care about the boundary condition
at infinity very much.
(ii)~To understand the proof of the nodal theorem~\cite{Amann:1995} is not easy,
but it is obvious that 
the existence of the regular $S$-deformation is a sufficient condition for the stability.
(iii)~We can easily show the non-existence of the zero mode just by showing two different regular 
$S$-deformation functions (see Appendix.~\ref{appendix:zeromode}).
In an usual way, we need to solve the 
Schr\"odinger equation for $E = 0$ with the decaying boundary condition at $x \to \infty$
or $x \to -\infty$, which is not always easy.

Finally, we briefly discuss an extension of the $S$-deformation method.
Let us consider a system 
\begin{align}
- \frac{d^2}{dx^2}\bm{\Phi} + \bm{V} \bm{\Phi} = \omega^2 \bm{{\cal C}} \bm{\Phi},
\label{extendedeq}
\end{align}
where $\bm{V}$ and $\bm{{\cal C}}$ are Hermitian matrices and 
$\bm{{\cal C}}$ is positive definite.
Since $\bm{\Phi}^\dag \bm{{\cal C}} \bm{\Phi} \ge 0$, from the similar discussion in Sec.~\ref{sec:2},
the existence of regular solution of $d\bm{S}/dx = \bm{S}^2 - \bm{V}$ implies the mode stability of
this system.

\section*{Acknowledgments}
We would like to thank B. Way for useful comments.
M.K. acknowledges financial support provided under
the European Union's H2020 ERC Consolidator Grant
``Matter and strong-field gravity: New frontiers in Einstein's theory''
grant agreement no. MaGRaTh-646597, 
and under the H2020-MSCA-RISE-2015 Grant No. StronGrHEP-690904.
T.T. acknowledges support in part by 
MEXT Grant-in-Aid for Scientific Research on Innovative Areas, Nos. 17H06357 and 17H06358, and by 
Grant-in-Aid for Scientific Research Nos. 26287044 and 15H02087.
M.K. also thanks Yukawa Institute for Theoretical Physics (YITP) at Kyoto University 
for their hospitality and the workshop ``Dynamics in Strong Gravity Universe'' YITP-T-18-05 held at YITP.

\appendix

\section{Hermitian $\bm{V}$ case as real symmetric problem}
\label{appendix:vermitianvasrealsymm}

We decompose $n\times n$ Hermitian potential $\bm{V} $ and $\bm{\Phi}$ into the real part and imaginary part 
as $\bm{V} = \bm{V}_R + i \bm{V}_I$ and $\bm{\Phi} = \bm{\Phi}_R + i \bm{\Phi}_I$.
From $\bm{V}^\dag = \bm{V}$, we obtain $\bm{V}_R^T = \bm{V}_R, \bm{V}_I^T = -\bm{V}_I$, where $T$
denotes a transposition.
The Schr\"odinger equation~\eqref{multischrodingereq} can be written in the form
\begin{align}
-\frac{d^2\bm{\Phi}_R}{dx^2} + \bm{V}_R \bm{\Phi}_R  -  \bm{V}_I \bm{\Phi}_I &= E \bm{\Phi}_R,
\\
- \frac{d^2\bm{\Phi}_I}{dx^2} +  \bm{V}_I \bm{\Phi}_R +  \bm{V}_R \bm{\Phi}_I &= E \bm{\Phi}_I.
\end{align}
So, we can consider this problem as $2n\times 2n$ real potential problem with
\begin{align}
\bm{U} = 
\begin{pmatrix}
\bm{V}_R  & -\bm{V}_I
\\ 
  \bm{V}_I & \bm{V}_R 
\end{pmatrix}
\end{align}
We should note that $\bm{U}$ is real symmetric since 
\begin{align}
\bm{U}^T = 
\begin{pmatrix}
\bm{V}_R^T  &  \bm{V}_I^T
\\ 
-\bm{V}_I^T  & \bm{V}_R^T
\end{pmatrix}
=
\begin{pmatrix}
\bm{V}_R  &- \bm{V}_I
\\ 
\bm{V}_I  & \bm{V}_R
\end{pmatrix} = \bm{U}.
\end{align}
When $\{\bm{\Phi}_\mu\} = \{\bm{\Phi}_{\mu R} + i \bm{\Phi}_{\mu I}\}, (\mu = 1,2,\cdots,m$ and $m\le 2n)$,
are linearly independent with complex coefficients,
\begin{align}
\begin{pmatrix}
\bm{\Phi}_{1R}
\\ 
\bm{\Phi}_{1I}  
\end{pmatrix},
\begin{pmatrix}
-\bm{\Phi}_{1I}
\\ 
\bm{\Phi}_{1R}
\end{pmatrix},
\begin{pmatrix}
 \bm{\Phi}_{2R}
\\ 
\bm{\Phi}_{2I}
\end{pmatrix},
\begin{pmatrix}
-\bm{\Phi}_{2I}
\\ 
 \bm{\Phi}_{2R} 
\end{pmatrix},
\cdots,
\begin{pmatrix}
 \bm{\Phi}_{mR}
\\ 
\bm{\Phi}_{mI}
\end{pmatrix},
\begin{pmatrix}
-\bm{\Phi}_{mI}
\\ 
 \bm{\Phi}_{mR} 
\end{pmatrix},
\end{align}
are linearly independent with real coefficients.
For $n$ linearly independent solutions $\{\bm{\Phi_i}\}$ ($i = 1,2\cdots, n$), 
defining $2n\times 2n$ real matrix
\begin{align}
\bm{Z} =
\begin{pmatrix}
\bm{\Phi}_{1R} &  \bm{\Phi}_{2R} & \cdots &\bm{\Phi}_{nR}
&-\bm{\Phi}_{1I} &  -\bm{\Phi}_{2I} & \cdots & -\bm{\Phi}_{nI}
\\ 
\bm{\Phi}_{1I}  &\bm{\Phi}_{2I}  &\cdots &\bm{\Phi}_{nI} 
& \bm{\Phi}_{1R} & \bm{\Phi}_{2R}  &\cdots & \bm{\Phi}_{nR}
\end{pmatrix},
\end{align}
this corresponds to Eq.~\eqref{ybardef} to the potential problem for $\bm{U}$.
We can show the relation
\begin{align}
|\det(\bm{Y})|^2 = \det(\bm{Z}), 
\end{align}
where $n\times n$ matrix $\bm{Y}$ is
\begin{align}
\bm{Y} =
(\bm{\Phi}_1,\bm{\Phi}_2,\cdots, \bm{\Phi}_n) &= 
(\bm{\Phi}_{1R} + i \bm{\Phi}_{1I},\bm{\Phi}_{2R} + i \bm{\Phi}_{2I},\cdots, \bm{\Phi}_{nR} + i \bm{\Phi}_{nI}).
\end{align}
Thus, $\det(\bm{Y}) = 0$ if and only if $\det(\bm{Z}) = 0$.

\section{Riccati transformation for a system of coupled Schr\"odinger equations}
\label{appendix:riccatitr}
We introduce the 
Riccati transformation for a system of coupled Schr\"odinger equations~\eqref{multischrodingereq}.
If there exists an $n\times n$ matrix $\bm{S}$ which satisfies 
$d\bm{\Phi}/dx = -\bm{S}\bm{\Phi}$, 
we can show
\begin{align}
\frac{d^2\bm{\Phi}}{dx^2}
&= 
-\frac{d\bm{S}}{dx} \bm{\Phi} + \bm{S}^2\bm{\Phi}.
\end{align}
We should note that $\bm{S}$ is not necessarily a Hermitian matrix in this section.
From the Schr\"odinger equation \eqref{multischrodingereq}, which $\bm{\Phi}$ is supposed to solve, we obtain
\begin{align}
\left(\bm{V} - E  + \frac{d\bm{S}}{dx} - \bm{S}^2\right)\bm{\Phi} = 0,
\end{align}
where we omitted to write an unit matrix in front of $E$.
If $\bm{\Phi}$ is a zero vector, the equation
\begin{align}
\det \left(\bm{V} - E  + \frac{d\bm{S}}{dx} - \bm{S}^2\right) = 0
\end{align}
should hold.
In fact, 
instead of solving the Schr\"odinger equation \eqref{multischrodingereq},
we can solve the equation
\begin{align}
\bm{V} - E  + \frac{d\bm{S}}{dx} - \bm{S}^2 = 0.
\label{riccatimulti}
\end{align}
Once a solution $\bm{S}$ is given, we can solve the equation $d\bm{\Phi}/dx = -\bm{S}\bm{\Phi}$
with respect to $\bm{\Phi}$, then $\bm{\Phi}$ satisfies the Schr\"odinger equation \eqref{multischrodingereq}.

On the other hand, for given $n$-linearly independent 
solutions $\{\bm{\Phi}_i \}$, $(i =1,2,\ldots,n)$ 
of Eq.\eqref{multischrodingereq},
we can construct $\bm{S}$ which satisfies Eq.~\eqref{riccatimulti} in the following way.
For an $n\times n$ matrix 
\begin{align}
\bm{Y} := (\bm{\Phi}_1, \bm{\Phi}_2,\ldots,\bm{\Phi}_n),
\label{ybardef}
\end{align}
which satisfies
\begin{align}
-\frac{d^2}{dx^2}\bm{Y} + \bm{V}\bm{Y}= E \bm{Y},
\label{schrodingereqY}
\end{align}
we define $\bm{S}$ as 
\begin{align}
\bm{S} := - \frac{d\bm{Y}}{dx} \bm{Y}^{-1},
\label{sfrombary}
\end{align}
where we consider a domain such that $\det(\bm{Y}) \neq 0$.
Since $\bm{S}$ satisfies $d\bm{Y}/dx = - \bm{S} \bm{Y}$, the relation
\begin{align}
\left(\bm{V} - E  + \frac{d\bm{S}}{dx} - \bm{S}^2\right)\bm{Y} = 0
\label{matrixschrodingereq}
\end{align}
holds. 
Since we focus on a domain with $\det(\bm{Y}) \neq 0$, this equation implies that 
$\bm{S}$ satisfies Eq.~\eqref{riccatimulti}.
For a solution of Eq.~\eqref{riccatimulti}, we can solve $d\bm{Y}/dx = - \bm{S} \bm{Y}$ with respect to $\bm{Y}$,
then it satisfies Eq.~\eqref{schrodingereqY}.
This solution $\bm{Y}$ contains $2 n^2$ integral constants 
which coincides with the number of integral constants of the general solution of Eq.~\eqref{schrodingereqY}.

We also discuss the condition for $\bm{S}$ defined by Eq.~\eqref{sfrombary}
to be a Hermitian matrix.
We can easily show that $\bm{S}$ becomes Hermitian
if it is chosen to be so at a point in solving Eq.~\eqref{sfrombary}.\footnote{From Eq.~\eqref{riccatimulti} and its Hermitian conjugate, 
an equation $d(\bm{S}^\dag - \bm{S})/dx = 
(\bm{S}^\dag - \bm{S})\bm{S}^\dag + \bm{S}(\bm{S}^\dag - \bm{S})$ holds.
If $\bm{S} = \bm{S}^\dag$ is satisfied at some point $x = x_0$, 
$\bm{S}$ becomes a Hermitian matrix. \label{foot8}}
This criterion is useful, but ``a point'' should be a finite point, cannot be infinity.
To discuss the relation between the decaying boundary condition at infinity and the Hermiticity of $\bm{S}$, 
here we introduce another criterion.
Defining $\bm{\rho}$ as
\begin{align}
\bm{\rho} := \bm{Y}^\dag \frac{d\bm{Y}}{dx}
-
\frac{d\bm{Y}^\dag}{dx}
\bm{Y},
\end{align}
one can show $d\bm{\rho}/dx = 0$ and hence
\begin{align}
\bm{\rho} = {\rm const}.
\end{align}
Thus, if $\bm{\rho} = 0$ at some point, $\bm{S}$ becomes a Hermitian matrix from the relation
\begin{align}
\bm{S}^\dag - \bm{S} = (\bm{Y}^\dag)^{-1} \bm{\rho} \bm{Y}^{-1}.
\end{align}
When all $\{\bm{\Phi}_i\}$ are decaying modes (or constant modes), 
$\bm{\rho}$ becomes zero since both $\bm{Y}$ and $d\bm{Y}/dx$ 
(only $d\bm{Y}/dx$ for constant modes) vanish at infinity,
and $\bm{S}$ becomes a Hermitian matrix.

\section{Existence of $S$-deformation of a compact support potential}
\label{appendix:compactpotential}

For simplicity, we consider the case with a compact support potential, {\it i.e.,}
$\bm{V} = 0$ for $|x| \ge L$ with a positive constant $L$.
Note that the following discussion can be extended to the case with a non-compact support potential 
if it rapidly decays at $x \to \pm \infty$.
Also, if $\bm{V} \to {\rm diag}[v_1,\ldots,v_n]$ with $v_i \ge 0, (i=1, \cdots, n)$ 
at $x \to \infty$ in some basis and vanishing components are rapidly decaying, 
the following discussion remains to be valid by a minor modification.\footnote{
Note that $\bm{V}$ at $x \to -\infty$  is zero because it corresponds to the horizon.
If $\bm{V}$ at $x \to \infty$ is not diagonalized, we can diagonalize it by a linear transformation of 
the wave functions.}

We denote $\bm{Y}_L^{(E)}$ and $\bm{Y}_R^{(E)}$ are solutions of Eq.~\eqref{schrodingereqY} for 
an energy $E (\le 0)$ which behave as
\begin{align}
\bm{Y}_L^{(E)}|_{x \le -L} &= {\rm diag}[e^{\sqrt{-E}x}, e^{\sqrt{-E}x}, \ldots, e^{\sqrt{-E}x}],
\\
\bm{Y}_R^{(E)}|_{x \ge L} &= {\rm diag}[e^{-\sqrt{-E}x}, e^{-\sqrt{-E}x}, \ldots, e^{-\sqrt{-E}x}].
\end{align}
For $E < 0$,
$\bm{Y}_L^{(E)}$ and $\bm{Y}_R^{(E)}$ are decaying mode at $x \to -\infty$ and $x \to \infty$, respectively.
It is clear that both $\bm{Y}_L^{(E)}$ and $\bm{Y}_R^{(E)}$ are continuous functions
of $E (\le 0)$ at a finite point of $x$.
We also define the corresponding $\bm{S}$ as 
\begin{align}
\bm{S}_L^{(E)} := -\frac{d\bm{Y}_L^{(E)}}{dx}\left(\bm{Y}_L^{(E)}\right)^{-1},
\quad
\bm{S}_R^{(E)} := -\frac{d\bm{Y}_R^{(E)}}{dx}\left(\bm{Y}_R^{(E)}\right)^{-1}.
\end{align}
In this section, we show that $\bm{S}_L^{(0)}$ and $\bm{S}_R^{(0)}$ are continuous for $-\infty < x <\infty$
when there exists no bound state with $E \le 0$.
Hereafter, we mainly focus on $\bm{S}_L^{(E)}$, but the same kind of discussion also holds for $\bm{S}_R^{(E)}$.

\begin{lemma}\label{lem1}
We denote Hermitian solutions of Eq.~\eqref{riccatimulti} for $E_1$ and $E_2$ with $E_2 > E_1$
by $\bm{S}^{(E_1)}$ and $\bm{S}^{(E_2)}$, respectively.
We assume that $\bm{S}^{(E_1)}$ and $\bm{S}^{(E_2)}$ are continuous for $x_- \le x \le x_+$.
If all eigenvalues of $\bm{S}^{(E_2)} - \bm{S}^{(E_1)}$ are positive at $x = x_-$,
those for $x_- \le x \le x_+$ are also positive.
\end{lemma}
\begin{proof}
Suppose that $\bm{S}^{(E_2)} - \bm{S}^{(E_1)}$ has a zero eigenvalue at 
$x = a$ with $x_- < a < x_+$, and it does not have a zero eigenvalue for $x_- \le x < a$.
At $x = a$, the equations
\begin{align}
\left(\bm{S}^{(E_2)} - \bm{S}^{(E_1)}\right) \hat{\bm{e}}_a &= 0,
\label{s2minuss1ata}
\\
\hat{\bm{e}}_a^\dag \left( \frac{d\bm{S}^{(E_2)}}{dx} - \frac{d\bm{S}^{(E_1)}}{dx}  \right)\hat{\bm{e}}_a &\le 0,
\label{se2se1inequality}
\end{align}
should hold,
where $\hat{\bm{e}}_a$ is an unit eigenvector
of $\bm{S}^{(E_1)} - \bm{S}^{(E_2)}$ 
for the zero eigenvalue at $x = a$.~\footnote{
Let us denote $\bm{e}$ be a unit eigenvector of 
$\bm{S}^{(E_2)} - \bm{S}^{(E_1)}$ with the eigenvalue $\Lambda$ which vanishes at $x = a$.
The derivative of $\Lambda$ becomes
$d\Lambda/dx =  \bm{e}^\dag (\bm{S}^{(E_2)} - \bm{S}^{(E_1)}) \bm{e}$,
where we used 
$(d\bm{e}^\dag/dx)(\bm{S}^{(E_2)} - \bm{S}^{(E_1)})\bm{e} +
\bm{e}^\dag(\bm{S}^{(E_2)} - \bm{S}^{(E_1)})(d\bm{e}/dx) = \Lambda ((d\bm{e}^\dag/dx)\bm{e} + \bm{e}^\dag (d\bm{e}/dx)) =0 $.
Since $\bm{e} = \hat{\bm{e}}_a$ at $x = a$, we obtain Eq.~\eqref{se2se1inequality} from $d\Lambda/dx|_{x = a} \le 0$.
}
We note that $\hat{\bm{e}}_a$ is a constant vector.\footnote{
In this paper, we use the hat symbol to denote a constant vector.
}
However, at $x = a$, we also have the relation
\begin{align}
\hat{\bm{e}}_a^\dag \left( \frac{d\bm{S}^{(E_2)}}{dx} - \frac{\bm{S}^{(E_1)}}{dx}  \right)\hat{\bm{e}}_a
&= 
E_2 - E_1 + \hat{\bm{e}}_a^\dag \left(\bm{S}^{(E_2)}\right)^2 \hat{\bm{e}}_a
- \hat{\bm{e}}_a \left(\bm{S}^{(E_1)}\right)^2 \hat{\bm{e}}_a
\notag\\ &= E_2 - E_1 > 0,
\end{align}
where we used Eq.~\eqref{s2minuss1ata} in the second equality.
This is a contradiction. 
\end{proof}

\begin{lemma}\label{lem2}
Let us assume that $\bm{S}$ is continuous for $x_- < x < x_+$ and $\bm{S}$ 
is divergent at $x = x_+$ $(x = x_-)$.
Then, the eigenvalue of $\bm{S}$ that diverges behaves as 
$-1/(x - x_+) \to +\infty$ at $x \to x_+ -0$ $(-1/(x - x_-) \to  -\infty$ at $x \to x_- + 0)$.
\end{lemma}
\begin{proof}
From the assumption, 
one of the eigenvalue of $\bm{S}$ is divergent at $x = x_+$.
Let $\bm{e}$ be the unit eigenvector of $\bm{S}$ with the eigenvalue $\lambda$
\begin{align}
\bm{S} \bm{e} = \lambda \bm{e},
\label{selambdae}
\end{align}
and $|\lambda| \to \infty$ at $x = x_+$.
Note that $\bm{e}$ is not a constant vector.
Taking $x$ derivative of Eq.~\eqref{selambdae} and multiplying $\bm{e}^\dag$ from the left side,
we obtain
\begin{align}
\bm{e}^\dag \frac{d\bm{S}}{dx} \bm{e}  &= 
\frac{d\lambda}{d x}.
\end{align}
From Eq.~\eqref{riccatimulti}, this equation can be written in the form
\begin{align}
\frac{d\lambda}{d x} 
&= 
\lambda^2
-\bm{e}^\dag(\bm{V} - E)\bm{e}.
\end{align}
Thus, when $|\lambda| \to \infty$ at $x \to x_+$, $\lambda$ approximately satisfies
$d\lambda/dx \simeq \lambda^2$, and the approximate solution is 
$\lambda \simeq - 1/(x - x_+) \to + \infty$ at $x \to x_+ -0$.\footnote{
Precisely speaking, this is just a sketch of the proof, 
but we can also give a rigorous proof.}
\end{proof}

\begin{lemma}\label{lem3}
Suppose that $\det(\bm{Y})$  is not identically zero.
$\bm{S}$ is divergent at $x = x_+$ if and only if $\det(\bm{Y}) = 0$ at $x = x_+$.
\end{lemma}
\begin{proof}
It is trivial that $\det(\bm{Y}) = 0$ at $x = x_+$ if $\bm{S}$
is divergent at $x = x_+$ because 
$\bm{S}= - (d\bm{Y}/dx) \bm{Y}^{-1}$, and $\bm{Y}$ and $d\bm{Y}/dx$ are continuous in $-\infty < x < \infty$.
Thus, we need to show that 
$\bm{S}$ is divergent at $x = x_+$ if $\det(\bm{Y}) = 0$ at $x = x_+$.
Let us assume that $\bm{S}$ is not divergent at 
$x = x_+$ when $\det(\bm{Y}) = 0$ at $x = x_+$.
In that case, from the relation
\begin{align}
\frac{d\det(\bm{Y})}{dx}
= \det(\bm{Y}) {\rm Tr}(\bm{Y}_L^{-1} d\bm{Y}/dx) 
= -\det(\bm{Y}) {\rm Tr}(\bm{S}),
\end{align}
$\det(\bm{Y})$ vanishes everywhere. This contradicts 
with the assumption that  $\det(\bm{Y})$  is not identically zero.
\end{proof}

\begin{lemma}\label{lem4}
If $\det\left(\bm{Y}_L^{(E_2)}\right)$ has a zero at $x = x_{E_2}$, 
there exists $E_1 (<E_2)$ such that
$\det\left(\bm{Y}_L^{(E)}\right)$ has a zero at $x = x_{E}$ for $E_1 < E < E_2$.
\end{lemma}
\begin{proof}
From lemma~\ref{lem3},
If $\det\left(\bm{Y}_L^{(E_2)}\right)$ has a zero at $x = x_{E_2}$, 
$\bm{S}_L^{(E_2)}$ is divergent at $x = x_{E_2}$.
From lemma~\ref{lem2}, one of the eigenvalue of $\bm{S}_L^{(E_2)}$, which is denoted by $\lambda^{(E_2)}$, behaves
like $(\lambda^{(E_2)})^{-1} \simeq - (x - x_{E_2})$ near $x = x_{E_2}$.
Since $(\lambda^{(E)})^{-1}$ is a continuous function of $E$, 
the property that $(\lambda^{(E)})^{-1}$ has a zero at $x = x_E$ with 
$d(\lambda^{(E)})^{-1}/dx|_{x_E} = -1$ should hold for the range $E_1 < E < E_2$ with some $E_1 (< E_2)$.
\end{proof}

\begin{lemma}\label{lem5}
Let us denote $x_E$ as a zero of $\det\left(\bm{Y}_L^{(E)}\right)$ if it exists
and assume that $\det\left(\bm{Y}_L^{(E)}\right) \neq 0$ for $-\infty < x < x_E$.
Then, $x_E$ is monotonically non-increasing function of $E (<0)$ if it exists.
\end{lemma}
\begin{proof}
Since at $x \le -L$, 
\begin{align}
\bm{S}_L^{(E)} &= -{\rm diag}[\sqrt{-E}, \sqrt{-E}, \ldots, \sqrt{-E}],
\end{align}
for $E_1 < E_2 < 0$, all eigenvalues of $\bm{S}_L^{(E_2)} - \bm{S}_L^{(E_1)}$ are positive at $x \le -L$.
If we assume $x_{E_1} < x_{E_2}$, $\bm{S}_L^{(E_2)}$ and $\bm{S}_L^{(E_1)}$ are continuous in
$-\infty < x < x_{E_1}$.
From lemma~\ref{lem1}, 
all eigenvalue of $\bm{S}_L^{(E_2)} - \bm{S}_L^{(E_1)}$ are positive for 
$-\infty < x < x_{E_1}$.
However, since
$\lambda_1^{(E_2)} > \lambda_1^{(E_1)}$ holds\footnote{
We used a well known fact:
{\it 
for two Hermitian matrices $\bm{A}$ and $\bm{B}$,
if all eigenvalues of $\bm{A} - \bm{B}$ are positive,
the inequality $\lambda_i^{A} > \lambda_i^B$ holds, where
$\lambda_i^A$ and $\lambda_i^B$ are the $i$-th largest eigenvalues of $\bm{A}$ and $\bm{B}$, respectively.
}}
and $\lambda_1^{(E_1)} \to \infty$ at $x \to x_{E_1} -0$, 
$\lambda_1^{(E_2)}$ is divergent at $x \to x_{E_1} -0$.
This contradicts with the assumption $x_{E_1} < x_{E_2}$.
\end{proof}

\begin{lemma}\label{lem6}
$\det\left(\bm{Y}_L^{(E)}\right)$ does not have any zero for sufficiently large negative $E$.
\end{lemma}
\begin{proof}
Let $\bm{e}$ be a unit eigenvector of $\bm{S}_L^{(E)}$ with the eigenvalue $\lambda^{(E)}_L$:
\begin{align}
\bm{S}_L^{(E)} \bm{e} &= \lambda^{(E)}_L \bm{e}.
\end{align}
For a sufficiently large negative $E$, the function 
$\bm{e}^\dag \bm{V} \bm{e} - E$ is positive and bounded everywhere.
In this case, from the equation
\begin{align}
\frac{d \lambda^{(E)}_L}{d x} &=
\left(\lambda^{(E)}_L - \sqrt{\bm{e}^\dag \bm{V} \bm{e} - E}\right)
\left(\lambda^{(E)}_L + \sqrt{\bm{e}^\dag \bm{V} \bm{e} - E}\right),
\label{eq:lambda:lem6}
\end{align}
we can see that $\lambda^{(E)}_L$ 
is bounded from the same discussion as~\cite{Kimura:2017uor}.~\footnote{
$\lambda^{(E)}_L$ behaves $\lambda^{(E)}_L|_{x \le -L} = -\sqrt{-E}$, so
we need to check that it is bounded for $x > -L$.
For positive $\bm{e}^\dag \bm{V} \bm{e} - E$, once 
$\lambda^{(E)}_L$ takes a value less than $\sqrt{\bm{e}^\dag \bm{V} \bm{e} - E}$
it is bounded in $x > -L$, and this condition is satisfied at $x = -L$.
}
Thus, the eigenvalue of $\bm{S}_L^{(E)}$ is continuous everywhere
for a large negative $E$,
hence $\det\left(\bm{Y}_L^{(E)}\right)$ also does not have any zero for a large negative $E$.
\end{proof}

\begin{proposition}
$\bm{S}_L^{(0)}$ is continuous if there exists no bound state with $E \le 0$.
\end{proposition}
\begin{proof}
To show that 
$\bm{S}_L^{(0)}$ is continuous for a stable spacetime,
we need to show that there exists a negative energy eigenstate if $\det\left(\bm{Y}_L^{(0)}\right)$ has a zero.
From lemmas~\ref{lem4}, ~\ref{lem5} and \ref{lem6},
there exists $E_0 (<0)$ such that $x_{E} \to \infty$ when $E \to E_0 + 0$.
Thus, there exists $E_1$ with $E_0 < E_1 < 0$ such that 
for $E_0 < E < E_1$, 
$x_E$, a zero of $\det\left(\bm{Y}_L^{(E)}\right)$, locates in $x \ge L$ and
$\det\left(\bm{Y}_L^{(E)}\right) \neq 0$ for $-\infty < x < x_E$.
For $x \ge L$, $\bm{Y}_L^{(E)}$ can be written in the form
\begin{align}
\bm{Y}_L^{(E)}|_{x \ge L} &= 
(e^{2\sqrt{-E}x} \bm{{\cal A}}_L^{(E)} +  1 ) e^{-\sqrt{-E}x} \bm{{\cal B}}_L^{(E)},
\end{align}
where $\bm{{\cal A}}_L^{(E)}$ and $\bm{{\cal B}}_L^{(E)}$ are constant matrices.
Note that $\det\left(\bm{{\cal B}}_L^{(E)}\right) \neq 0$ because $\det\left(\bm{Y}_L^{(E)}\right)\Big|_{x > L}$ is not identically zero.
If $\det(\bm{Y}_L^{(E)} )= 0$ at $x = x_E$, the equation 
\begin{align}
\det\left( e^{2\sqrt{-E}x_E} \bm{{\cal A}}_L^{(E)} +  1 \right)= 0,
\end{align}
holds.
Since $x_E \to \infty$ for $E \to E_0$, 
$\det\left(\bm{{\cal A}}_L^{(E_0)}\right) = 0$ should be satisfied.
Let $\bm{u}$ be an eigenvector of $\bm{{\cal A}}_L^{(E_0)}$ with zero eigenvalue.
$\bm{Y}_L^{(E_0)}\left(\bm{{\cal B}}_L^{(E_0)}\right)^{-1} \bm{u}$ behaves like
\begin{align}
\bm{Y}_L^{(E_0)} \left(\bm{{\cal B}}_L^{(E_0)}\right)^{-1} \bm{u}\Big|_{x \le -L} &= 
e^{\sqrt{-E_0}x} \left(\bm{{\cal B}}_L^{(E_0)}\right)^{-1} \bm{u},
\\
\bm{Y}_L^{(E_0)} \left(\bm{{\cal B}}_L^{(E_0)}\right)^{-1} \bm{u}\Big|_{x \ge L} &= 
e^{-\sqrt{-E_0}x} \bm{u}.
\end{align}
Thus, $\bm{Y}_L^{(E_0)} \left(\bm{{\cal B}}_L^{(E_0)}\right)^{-1} \bm{u}$ 
is a negative energy eigenstate of the Schr\"odinger equation.
\end{proof}

\section{Robustness of the $S$-deformation method}
\label{appendix:robustness}

In a single mode case, 
the robustness of the $S$-deformation method was discussed in~\cite{Kimura:2018eiv},
{\it i.e.,} the general regular $S$-deformation has $1$ parameter degree of freedom.
In this section, we extend the discussion in~\cite{Kimura:2018eiv} to coupled systems when the 
potential $\bm{V}$ has a compact support.

In this section, since we only consider $E = 0$ case, we omit to write the superscript $(0)$, {\it e.g.,}
$\bm{Y}_L$ instead of $\bm{Y}_L^{(0)}$.
We assume that there exists no bound state with $E \le 0$, then 
$\det(\bm{Y}_L)$ and $\det(\bm{Y}_R)$ do not have zero in $-\infty < x < \infty$.
The general solution of Eq.~\eqref{schrodingereqY} with $E = 0$ is
\begin{align}
\bm{Y} = (\bm{Y}_L \bm{M} + \bm{Y}_R) \bm{N}.
\end{align}
Since we are interested in regular $S$-deformations which correspond to
the case with $\det(\bm{Y}) \neq 0$, 
we only need to consider $\det(\bm{N}) \neq 0$.
In that case, all $\bm{N}$ gives the same $\bm{S} = -(d\bm{Y}/dx)\bm{Y}^{-1}$.
Thus, it is enough to study the case that $\bm{N}$ is the identity matrix.

\subsection{All eigenvalues of $\bm{S}_R  - \bm{S}_L$ are positive}
In $x \le -L$, from the behavior of $\bm{Y}_L$ and $\bm{Y}_R$,
\begin{align}
\bm{Y}_L|_{x \le -L} &= {\rm diag}[1, \ldots, 1],
\\
\bm{Y}_R|_{x \le - L} &= 
\bm{{\cal C}}_R +  x \bm{{\cal D}}_R,
\end{align}
$\bm{S}_L$ and $\bm{S}_R$ becomes
\begin{align}
\bm{S}_L |_{x \le -L} &= {\rm diag}[0,\ldots,0],
\\
\bm{S}_R |_{x \le -L} &= - (\bm{{\cal D}}_R^{-1}\bm{{\cal C}}_R +  x )^{-1}.
\end{align}
Note that $\det(\bm{{\cal D}}_R) \neq 0$ because we assume that there exists no bound state with $E \le 0$.
Since $\bm{S}_R$ is bounded, 
$\bm{{\cal D}}_R^{-1}\bm{{\cal C}}_R +  x$ does not have zero eigenvalue in $x \le -L$,
hence $\bm{{\cal D}}_R^{-1}\bm{{\cal C}}_R  - L$ is negative definite 
and $\bm{S}_R |_{x < -L}$ is positive definite in $x \le -L$.
Thus, all eigenvalues of $\bm{S}_R - \bm{S}_L$ are positive in $x \le -L$.

Suppose that $\bm{S}_R - \bm{S}_L$ has a zero eigenvalue at $x = a (> -L)$.
Also, let $\bm{e}$ be a unit eigenvector of $\bm{S}_R - \bm{S}_L$ with the eigenvalue $\Lambda$,
\begin{align}
(\bm{S}_R - \bm{S}_L) \bm{e} = \Lambda \bm{e},
\end{align}
then $\Lambda = 0$ at $x = a$.
From this equation, we can derive
\begin{align}
\frac{d\Lambda}{dx} &= \Lambda \left(
\Lambda 
+
2 \bm{e}^\dag \bm{S}_L \bm{e}
\right).
\end{align}
Thus, if $\Lambda = 0$ at $x = a$, $\Lambda = 0$ everywhere 
from the uniqueness of the ordinary differential equation.
However, this contradicts with the fact that $\Lambda > 0$ in $x \le -L$ as shown above.

\subsection{general regular $\bm{S}$}
Defining $\bm{\rho}_{LR}$ by
\begin{align}
\bm{\rho}_{LR} = 
\frac{d\bm{Y}_L^\dag}{dx} \bm{Y}_R - 
 \bm{Y}^\dag_L \frac{d\bm{Y}_R}{dx},
\end{align}
the relation
\begin{align}
\bm{S}_R - \bm{S}_L = (\bm{Y}^\dag_L)^{-1} \bm{\rho}_{LR} (\bm{Y}_R )^{-1} =: \bm{W}^{-1},
\end{align}
holds. From $\bm{Y} = \bm{Y}_L \bm{M} + \bm{Y}_R$,
\begin{align}
\bm{W}_L := 
\bm{Y} \bm{\rho}_{LR}^{-1}\bm{Y}_L^\dag &= 
\bm{Y}_L \bm{M} \bm{\rho}_{LR}^{-1} \bm{Y}_L^\dag  
+ \bm{Y}_R \bm{\rho}_{LR}^{-1} \bm{Y}_L^\dag
\notag\\&=
\bm{Y}_L \bm{M} \bm{\rho}_{LR}^{-1} \bm{Y}_L^\dag  
+ \bm{W}.
\label{eq:wl}
\end{align}
From the relation, 
\begin{align}
\frac{d\bm{Y}_L^\dag}{dx} \bm{Y} - 
 \bm{Y}_L^\dag \frac{d\bm{Y}}{dx} 
&= 
\bm{\rho}_{LR},
\end{align}
where we use $(d\bm{Y}_L^\dag/dx)\bm{Y}_L - \bm{Y}_L^\dag (d\bm{Y}_L/dx) = 0$, 
we can show
\begin{align}
\bm{S}  -\bm{S}_L 
&= 
(\bm{Y}_L^\dag)^{-1}\bm{\rho}_{LR} \bm{Y}^{-1} = \bm{W}_L^{-1}.
\end{align}
Since $\bm{S}$ is Hermitian, $\bm{W}_L$ should also be Hermitian.
From Eq.~\eqref{eq:wl}, this condition is that 
$\bm{M} \bm{\rho}_{LR}^{-1}$ is Hermitian.
If $\bm{M} \bm{\rho}_{LR}^{-1}$ is non-negative definite,
from Eq.~\eqref{eq:wl}, 
$\bm{Y} \bm{\rho}_{LR}^{-1}\bm{Y}_L^\dag$ becomes positive definite.
In that case, $\bm{Y}$ does not have a zero eigenvalue 
because if $\bm{e}_1^\dag \bm{Y} = 0$ with a vector $\bm{e}_1$ at some point, it contradicts with
$\bm{e}_1^\dag \bm{Y} \bm{\rho}_{LR}^{-1}\bm{Y}_L^\dag \bm{e}_1 > 0$.
Thus, $\bm{M} \bm{\rho}_{LR}^{-1} \ge 0$ is the sufficient condition so that
$\bm{S}$ is regular everywhere.

From the relations 
$\bm{Y}_L \bm{M} \bm{\rho}_{LR}^{-1} \bm{Y}_L^\dag = {\cal O}(x^2)$ 
and $\bm{W} = {\cal O}(x)$ at $x \to \infty$,
and
$\bm{Y}_L \bm{M} \bm{\rho}_{LR}^{-1} \bm{Y}_L^\dag = {\cal O}(x^{0})$ 
and $\bm{W} = {\cal O}(x)$ at $x \to -\infty$,
if $\bm{M} \bm{\rho}_{LR}^{-1}$ has a negative eigenvalue, 
$\bm{W}_L$ has a negative eigenvalue near $x \to \infty$ and 
$\bm{W}_L$ is positive definite near $x \to -\infty$.
Thus, $\bm{W}_L$ should has a zero eigenvalue with an eigenvector $\bm{e}$ somewhere.
However, this implies $\bm{Y} (\bm{\rho}_{LR}^{-1}\bm{Y}_L^\dag \bm{e}) = 0$,
and then $\bm{Y}$ has a zero eigenvalue. This is a contradiction.
Thus, $\bm{M} \bm{\rho}_{LR}^{-1} \ge 0$ is the necessarily and 
sufficient condition so that $\bm{S}$ is regular everywhere.

\subsection{Robustness of the $S$-deformation method}
\label{appendix:sub:robustness}

We have already obtained that all eigenvalues of $\bm{S} - \bm{S}_L$ are positive in the previous subsection.
We repeat the same discussion for $\bm{S}_R - \bm{S}$.
From the equation $\bm{Y}^\dag = \bm{M}^\dag \bm{Y}_L^\dag + \bm{Y}_R^\dag$, the relation
\begin{align}
\bm{W}_R := 
\bm{Y}_R \bm{\rho}_{LR}^{-1} (\bm{M}^\dag)^{-1} \bm{Y}^\dag &= 
\bm{Y}_R \bm{\rho}_{LR}^{-1} \bm{Y}_L^\dag
+
\bm{Y}_R \bm{\rho}_{LR}^{-1} (\bm{M}^\dag)^{-1} 
\bm{Y}_R^\dag
\notag\\&= 
\bm{W}
+
\bm{Y}_R \bm{\rho}_{LR}^{-1} (\bm{M}^\dag)^{-1} 
\bm{Y}_R^\dag,
\end{align}
holds.
From the relation, 
\begin{align}
\frac{d\bm{Y}^\dag}{dx} \bm{Y}_R - 
 \bm{Y}^\dag \frac{d\bm{Y}_R}{dx} 
&= 
\bm{M}^\dag \bm{\rho}_{LR},
\end{align} 
we can show
\begin{align}
\bm{S}_R  -\bm{S}
&= 
(\bm{Y}^\dag)^{-1}\bm{M}^\dag \bm{\rho}_{LR} \bm{Y}_R^{-1} = \bm{W}_R^{-1}.
\end{align}
From the same discussion above, we impose 
that $\bm{M}^\dag \bm{\rho}_{LR}$ is Hermitian
and $\bm{M}^\dag \bm{\rho}_{LR}$ is positive definite\footnote{
This is same as that $\bm{M} \bm{\rho}_{LR}^{-1}$ is Hermitian
and positive definite.} so that $\bm{S}$ is Hermitian and continuous.
Thus, $\bm{S}_R - \bm{S}$ is positive definite everywhere for regular $\bm{S}$.

If we solve Eq.~\eqref{sdefeqmulti} with the boundary condition
such that all eigenvalues of $\bm{S}_R - \bm{S}$ and $\bm{S} - \bm{S}_L$ are positive,
this relations hold everywhere.
In that case, all eigenvalues satisfy $\lambda_i^{L} < \lambda_i < \lambda_i^{R}$, where
$\lambda_i^{L}, \lambda_i, \lambda_i^{R}$ are $i$-th largest eigenvalue of 
$\bm{S}_L, \bm{S}, \bm{S}_R$, respectively.
This is the reason why we can find the $S$-deformation without fine-tuning 
in numerically solving Eq.~\eqref{sdefeqmulti}.

\section{regular $\bm{S}$ for positive definite $\bm{V}$}
\label{appendix:positivedefinite}

In this section, we consider the case with positive definite $\bm{V}$, where the system is 
manifestly stable. In this case, the following proposition holds.

\begin{proposition}\label{SforpositivedefiniteV}
Let us assume that $\bm{V}$ is positive definite,
and let $\bm{S}$ be a solution of Eq.~\eqref{sdefeqmulti}.
If all components of $\bm{S}$ are zero at a point,
$\bm{S}$ is Hermitian and bounded everywhere.
\end{proposition}
\begin{proof}
Suppose that $x = a$ be a point such that all components of $\bm{S}$ are zero there,
then all eigenvalues are also zero at $x = a$.
As shown in footnote~\ref{foot8}, if $\bm{S}$ is Hermitian at a point, 
the solution of Eq.~\eqref{sdefeqmulti} is Hermitian everywhere.
Let $\bm{e}$ and $\lambda$ be a unit eigenvector and an eigenvalue of $\bm{S}$, respectively.
Then, the relation
\begin{align}
\frac{d\lambda}{dx} = \lambda^2 - \bm{e}^\dag \bm{V} \bm{e}
\label{eq:lambdasdef}
\end{align}
holds and $\lambda = 0$ at $x = a$.
From the same discussion as~\cite{Kimura:2017uor},
once $\lambda = 0 $ at a point,
$\lambda$ is bounded above and below everywhere for positive $\bm{e}^\dag \bm{V} \bm{e}$.
Since this discussion holds for all eigenvalues of $\bm{S}$, 
the solution of Eq.~\eqref{sdefeqmulti} is bounded everywhere.
\end{proof}

The discussion in Appendix.~\ref{appendix:robustness} suggests that
all eigenvalues of $\bm{S}_R - \bm{S}$ and $\bm{S} - \bm{S}_L$ are positive everywhere for regular $\bm{S}$.
Thus, from proposition.~\ref{SforpositivedefiniteV} we can expect that
$\bm{S}_R$ is positive definite, and $\bm{S}_L$ is negative definite for positive definite $\bm{V}$.

If the potential is positive definite only at large $x$, 
the similar discussion holds in the asymptotic region, hence
we can expect that the solution of Eq.~\eqref{sdefeqmulti} 
with the boundary condition $\bm{S} = 0$ at large $x$ becomes regular for a stable spacetime.
This is just a rough discussion, but it is worth to try this boundary condition 
if $\bm{V}$ is positive definite at large $x$.
We should note that the numerical studies in Sec.~\ref{sec:application} support this.

\section{approximate solution in $n = 2$}
\label{approximatesolutionsdef}
Usually, the potential $\bm{V}$ is proportional to $r - r_H$ near the horizon of a non-extremal black hole, 
where $r_H$ is the horizon radius.
Since the relation between $x$ and $r$ near the horizon is $r/r_H \simeq 1 + e^{x/r_H}$,
the potential is rapidly decaying to zero at $x \to -\infty$ as 
$\bm{V} \propto  e^{x/r_H} \to 0$. 
In this case, we can find approximate solutions of the equation $\bm{V} + d{\bm S}/dx - \bm{S}^2= 0$ 
as~\footnote{When $\bm{V} = 0$, the general solution of $d^2\bm{Y}/dx^2 = 0$ is
$\bm{Y} = \bm{R}_{\Theta} (x + \bm{{\cal A}})\bm{{\cal B}}$ with constant matrices 
$\bm{R}_\Theta, \bm{{\cal A}}$ and $\bm{{\cal B}}$ ($\det(\bm{{\cal B}}) \neq 0$). 
For the latter convenience, we put the rotation matrix 
$\bm{R}_\Theta$
\begin{align}
\bm{R}_\Theta = 
\begin{pmatrix}
\cos \Theta  & -\sin \Theta\\
\sin \Theta             & \cos \Theta
\end{pmatrix}.
\end{align}
We should note that $\bm{{\cal B}}$ does not appear in $\bm{S} = - (d\bm{Y}/dx) \bm{Y}^{-1}$.
Imposing $\bm{S}$ to be symmetric, $\bm{{\cal A}}$ also should be 
symmetric. Since the non-diagonal part of $\bm{{\cal A}}$ can be removed by 
the degrees of freedom of $\bm{R}_\Theta$ and $\bm{{\cal B}}$, 
{\it i.e.,} $\bm{R}_\Theta \to \bm{R}_\Theta \bm{R}_{\Theta_1}$ and $\bm{{\cal B}} \to 
\bm{R}_{\Theta_1}^{-1} \bm{{\cal B}}$ and choose $\Theta_1$ so that the non-diagonal part vanishes, 
we only need to consider $\bm{Y} = \bm{R}_\Theta {\rm diag}[x - c_+, x - c_-]\bm{{\cal B}}$.
Defining $\theta_0 = 2\Theta$, 
we obtain Eqs.~\eqref{s11approx}-\eqref{s22approx} from $\bm{S} = - (d\bm{Y}/dx) \bm{Y}^{-1}$.
}
\begin{align}
S_{11} &= \frac{-2x + c_+ + c_- - (c_+ - c_-)\cos\theta_0}{2(x - c_+)(x - c_-)},
\label{s11approx}
\\
S_{12} &= \frac{-(x_+ - x_-)\sin \theta_0}{2(x - c_+)(x - c_-)},
\label{s12approx}
\\
S_{22} &= \frac{-2x + c_+ + c_- + (c_+ - c_-)\cos\theta_0}{2(x - c_+)(x - c_-)}.
\label{s22approx}
\end{align}
In $x \to - \infty$, we can see 
$S_{11} \simeq -1/x, S_{12} \simeq (-(x_+ - x_-)\sin \theta_0)/(2x^2)$ and $S_{22} \simeq -1/x$.
These equations can be written in the form
\begin{align}
c_\pm &= x - \frac{S_{11} + S_{22} \pm \sqrt{4 S_{12}^2 + (S_{11} - S_{22})^2}}{2(S_{12}^2 - S_{11} S_{22})},
\label{eq:cpm}
\\
\cos\theta_0 &= \frac{-S_{11} + S_{22}}{\sqrt{4 S_{12}^2 + (S_{11} - S_{22})^2}}.
\label{eq:thetazero}
\end{align}
For a numerical solution $\bm{S}$,
if the right hand sides of the above equations take constants in the asymptotic region, 
it implies that Eqs.~\eqref{s11approx}-\eqref{s22approx} become a good approximation there.
Also, if $x < c_\pm$ are satisfied in the asymptotic region, 
it implies that $\bm{S}$ is bounded in the asymptotic region.

\section{non-existence of the zero mode}
\label{appendix:zeromode}

Similarly to the case of the single degree of freedom~\cite{Kimura:2018eiv}, the following 
proposition on the zero mode, {\it i.e.,} the non-trivial solution of Eq.~\eqref{multischrodingereq} with zero energy
which is decaying (or constant) at both $x\to -\infty$ and $x\to \infty$,
holds:
\begin{proposition}\label{zeromode}
Suppose that there exist two different regular solutions of Eq.~\eqref{sdefeqmulti}, $\bm{S}_1$ and $\bm{S}_2$.
If $\det(\bm{S}_1 - \bm{S}_2)$ is not identically zero, no zero mode exists.
\end{proposition}
\begin{proof}
If there exists a zero mode $\bm{\Phi}$, it satisfies
\begin{align}
\frac{d\bm{\Phi}}{dx} + \bm{S}_i\bm{\Phi} = 0, ~ (i = 1, 2)
\label{zeroenergyphi}
\end{align}
from Eq.~\eqref{multisdef1},
and hence, $(\bm{S}_1 - \bm{S}_2)\bm{\Phi}=0$ holds.
However, from the assumption this cannot be satisfied.
\end{proof}

If we obtain two different $\bm{S}$ from the boundary conditions $\bm{S} = 0$ at two large $x$,
usually, the assumption ``$\det(\bm{S}_1 - \bm{S}_2)$ is not identically zero''
is satisfied.

\end{document}